\documentclass[conference,onecolumn,a4paper]{IEEEtran}

\usepackage{silence}
\usepackage{premable_doc}
\usepackage{premable_tikz}

\begin{document}

\title{Data-Driven Neural Polar Codes for Unknown Channels With and Without Memory}
\author{
\IEEEauthorblockN{Ziv Aharoni} 
    \IEEEauthorblockA{
    Ben-Gurion University\\
    zivah@post.bgu.ac.il }
\and
\IEEEauthorblockN{Bashar Huleihel} 
    \IEEEauthorblockA{
    Ben-Gurion University \\
    basharh@post.bgu.ac.il }
\and
\IEEEauthorblockN{Henry D. Pfister} 
    \IEEEauthorblockA{
    Duke University \\
    henry.pfister@duke.edu }
\and
\IEEEauthorblockN{Haim H. Permuter} 
    \IEEEauthorblockA{
    Ben-Gurion University \\
    haimp@bgu.ac.il }
}

\maketitle

\begin{abstract}
In this work, a novel data-driven methodology for designing polar codes for channels with and without memory is proposed. The methodology is suitable for the case where the channel is given as a "black-box" and the designer has access to the channel for generating observations of its inputs and outputs, but does not have access to the explicit channel model. 
The proposed method leverages the structure of the \gls{sc} decoder to devise a \gls{nsc} decoder. The \gls{nsc} decoder uses \glspl{nn} to replace the core elements of the original \gls{sc} decoder, the check-node, the bit-node and the soft decision. Along with the \gls{nsc}, we devise additional \gls{nn} that embeds the channel outputs into the input space of the \gls{sc} decoder. The proposed method is supported by theoretical guarantees that include the consistency of the \gls{nsc}. Also, the \gls{nsc} has computational complexity that does not grow with the channel memory size. This sets its main advantage over \gls{sct} decoder for \glspl{fsc} that has complexity of $O(|\cS|^3 N\log N)$, where $|\cS|$ denotes the number of channel states.
We demonstrate the performance of the proposed algorithms on memoryless channels and on channels with memory. The empirical results are compared with the optimal polar decoder, given by the \gls{sc} and \gls{sct} decoders. 
We further show that our algorithms are applicable for the case where there \gls{sc} and \gls{sct} decoders are not applicable. 

\end{abstract}


\glsresetall

\blfootnote{This paper was presented in part at the International Symposium on Information Theory (ISIT) 2023 \cite{aharoni2023data}.}
\section{Introduction}

\par Polar codes allow the construction of capacity-achieving codes for symmetric binary-input memoryless channels~\cite{arikan2009channel}. When given $N$ independent copies of a binary discrete memoryless channel (DMC) $W$, \gls{sc} decoding induces a new set of $N$ binary effective channels $W^{(i)}_N$. Channel polarization is the phenomenon whereby, for $N$ sufficiently large, almost all of the effective bit channels $W^{(i)}_N$ have capacities close to 0 or 1. Specifically, the fraction of channels with capacity close to 1 approaches $\sI(W)$ and the fraction of channels with capacity close to 0 approaches $1-\sI(W)$, where $\sI(W)$ is the channel's symmetric capacity. The construction of polar codes involves choosing which rows to keep from the square generator matrix given by Arikan's transform~\cite[Section VII]{arikan2009channel}. The encoding and decoding procedures are performed by recursive formulas whose computational complexity is $O(N\log N)$.

\par Polar codes can also be applied to \glspl{fsc}. Arikan's transform also polarizes the bit channels $W^{(i)}_N$ in the presence of memory \cite{csacsouglu2019polar}, and thus the encoding algorithm is the same as if the channel is memoryless. However, the decoding algorithm needs to be updated since the derivation of the \gls{sc} decoder in \cite{arikan2009channel} relies on the memoryless property. To account for the channel memory, the channel outputs are represented by a trellis, whose nodes capture the information of the channel's memory. This trellis was embedded into the \gls{sc} decoding algorithm to yield the \gls{sct} decoding algorithm \cite{wang2014joint, wang2015construction}.


\par However, the \gls{sct} decoder is only applicable when the channel model is known and when the channel's state alphabet size is finite and relatively small. For \glspl{fsc}, the computational complexity of the \gls{sct} decoder is $O(|\cS|^3 N \log N)$, where $|\cS|$ is the number of channel states.  For Markov channels where the set of channel states is not finite, the \gls{sct} decoder is not applicable without quantization of its states. With quantization, there may be a strong tension between the computational complexity and the error introduced by quantization. Additionally, the \gls{sct} decoder cannot be used for an unknown channel with memory without first estimating the channel as it requires an explicit channel model,.

\par The \gls{sct} decoder can also be applied to a larger class of channels (e.g., insertion and deletion channels) where, given the the channel output sequence $Y^N$, a trellis can be constructed to efficiently represent $P_{X^N,Y^N}$ \cite{tal2021polar}. In that case, the decoding complexity is upper bounded by $O(M^3 N \log N)$, where $M$ is the maximum number of states in any trellis stage.  If $M$ grows linearly with $N$, then the complexity of the decoder may grow very rapidly (e.g., $\geq N^{4}$) and is dominated by the number of trellis states rather than the block length. 

\subsection{Contribution}
\par We propose a novel methodology for data-driven polar decoders. The methodology treats the channel as a ``black-box'' used to generate samples of input-output pairs without an access to the channel's explicit model.
It dissects the polar decoder into two separate components. The first is the sufficient statistic of the channel outputs, that we denote by $E$. The function $E:\cY\to\cE$ embeds the channel outputs into a latent space $\cE\subset\bR^d$. The embeddings $e\in\cE$ are then used as the inputs of the second component - a \gls{nsc} decoder. 

\par The \gls{nsc} uses three \glspl{nn} that replace the three core elements of the \gls{sc} decoder: the check-node, the bit-node and the soft decision operations. The parameters of the embedding function $E$, and the \gls{nsc} parameters are determined in a training phase, in which the \gls{mi} of the effective channels $W^{(i)}_N$ is estimated. The training of the embedding and the \gls{nsc} are performed in two alternative ways. The first trains the embedding and the \gls{nsc} jointly. The second determines the parameters of the embedding $E$ using neural estimation methods \cite{tsur2022neural, tsur2022optimizing, aharoni2022density}, and then, determines the parameters of the \gls{nsc} while the parameters of $E$ are fixed. After the training phase, the set of ``clean'' effective channels are determined by a \gls{mc} evaluation of the \gls{mi} of the effective bit channels to complete the code design.

\par The \gls{nsc} decoder is a consistent estimator of an optimal polar decoder. Specifically, for \glspl{fsc}, the \gls{nsc} decoder provides a consistent estimator of the mutual information of the effective bit channels. We also show its computational complexity, and show it does not grow with the channel memory. This is a main advantage over the \gls{sct} decoder whose computational complexity grows cubicly with the channel memory size. To complete the methodology, we provide an extension of the \gls{nsc} for a stationary input distribution. This involves using the Honda-Yamamoto scheme \cite{honda2013polar} and incorporating it into our algorithms. 

\subsection{Related work}
\par The usage of \glspl{nn} for polar codes design were considered in the past. In \cite{doan2018neural}, \glspl{nn} were used to decrease the decoding latency by designing a \gls{nn} decoder that decodes multiple symbols at once. Other instances used \glspl{nn} to aid existing algorithms, such as \cite{xu2020deep, ebada2019deep, xu2017improved}. 
The paper \cite{makkuva2021ko} presents KO codes, a family of deep-learning driven codes Reed-Muller and Polar codes on the \gls{awgn} channel. KO codes are related to the methods proposed here in the sense that we also leverage the structure of the Arikan's transform to design efficient decoders. However, we do not change Arikan's transform, and we consider channels with memory.
To the best of our knowledge, there is no instance of a data-driven polar code design for channels with memory. This work aims to address this gap by developing the necessary algorithms for this task.

\subsection{Organization}
\par The paper is organized as follows. Section \ref{sec:pre} defines the notation and gives the necessary background on polar codes. Section \ref{sec:data_driven_polar_code:memoryless} presents the methodology for data-driven polar code design for memoryless channels. Section \ref{sec:data_driven_polar_code:memory} extends the methodology for the case where the channel has memory. Section \ref{sec:hy} extends the algorithms in Sections \ref{sec:data_driven_polar_code:memoryless} and \ref{sec:data_driven_polar_code:memory} to stationary input distributions. Section \ref{sec:exp} presented the numerical results. The proofs of the paper appear in Section \ref{sec:proofs}. Section \ref{sec:conc} includes conclusions and future goals. 

\section{Notations and Preliminaries} \label{sec:pre}
Throughout this paper, we denote by $(\Omega,\cF,\bP)$ the underlying probability space on which all random variables are defined, with $\bE$ denoting expectation. \Glspl{rv} are denoted by capital letters and their realizations will be denoted by lower-case letters, e.g. $X$ and $x$, respectively. Calligraphic letters denote sets, e.g.  $\mathcal{X}$. We use the notation $X^n$ to denote the \gls{rv} $(X_1,X_2,\dots,X_n)$ and $x^n$ to denote its realization. The probability $\Pr[X=x]$ is denoted by $P_X(x)$. Stochastic processes are denoted by blackboard bold letters, e.g., $\bX := (X_i)_{i\in\bN}$. An $n$-coordinate projection of $\bP$ is denoted by $P_{X^n Y^n}:=\bP\big|_{\sigma(X^n,Y^n)}$, where $\sigma(X^n,Y^n)$ is the $\sigma$-algebra generated by $(X^n,Y^n)$. We denote by $[N]$ the set of integers $\{1,\dots,N\}$.

\par The \gls{mi} between two \glspl{rv} $X,Y$ is denoted by $\f{\sI}{X;Y}$. The \gls{di} between $X^n$ and $Y^n$ is defined as $\mathsf{I}(X^n\rightarrow Y^n) = \sum_{i=1}^n \sI(X^i;Y_i|Y^{i-1})$ \cite{massey1990causality}. For two distributions $P,Q$, the \gls{ce} is denoted by $\ce{P}{Q}$, the entropy is denoted by $\f{\sH}{P}$ and the \gls{kl} divergence is denoted by $\kl{P}{Q}$. The notation $P\ll Q$ indicates that $P$ is absolutely continuous \gls{wrt} $Q$.

\par The tuple $\left(W_{Y\vert X},\cX,\cY\right)$ defines a memoryless channel with input alphabet $\cX$, output alphabet $\cY$ and a transition kernel $W_{Y\vert X}$. Throughout the paper we assume that $\cX=\left\{0, 1\right\}$. For a memoryless channel, we denote its input distribution by $P_X = P_{X_i}$ for all $i\in\bZ$. The tuple $\left(W_{Y\Vert X},\cX,\cY\right)$ defines a time invariant channel with memory, where  $W_{Y\Vert X}=\left\{W_{Y_0\vert Y^{-1}_{-i+1},X^0_{-i+1}}\right\}_{i\in\bN}$. The term $W_{Y^N\Vert X^N}=\prod_{i=1}^{N}W_{Y_0\vert Y^{-1}_{-i+1},X^0_{-i+1}}$ denotes the probability of observing $Y^N$ causally conditioned on $X^N$ \cite{kramer1998directed}. The symmetric capacity of a channel is denoted by $\f{\sI}{W}$. We denote by $\cD_{M,N}=\left\{x_{j,i},y_{j,i}\right\}_{j\in[M], i\in[N]} \sim P_{X^{MN}} \otimes W_{Y^{MN}\| X^{MN}}$ a finite sample of inputs-outputs pairs of $M$ consecutive blocks of $N$ symbols, where $x_{j,i}, y_{j,i}$ denotes the $i$-th input and output of the $j$-th block. The term $\cD_{MN}$ denotes the same sample after its concatenation into one long sequence of inputs and outputs pairs.

\subsection{Finite State Channels}\label{sec:pre:fsc}
\par A \gls{fsc} is defined by the tuple $(\mathcal{X},\cY,\mathcal{S},P_{S^\prime,Y|X,S})$, where $X$ is the channel input, $Y$ is the channel output, $S$ is the channel state at the beginning of the transmission, and $S^\prime$ is the channel state at the end of the transmission.
The cardinalities $\mathcal{X},\mathcal{S}$ are assumed to be finite.
At each time $t$, the channel has the Markov property, that is, $P_{S_t,Y_t|X^t,S^{t-1},Y^{t-1}} = P_{Y_t,S_t|X_t,S_{t-1}}$.
A \gls{fsc} is called indecomposable if for every $\varepsilon>0$ there exists an $n_0\in\bN$ such that for $n\ge n_0$ we have $\sup_{s_0, s_0^\prime, s_n\in\cS, x^n\in \cX^n} \vert \f{P_{S_n|X^n,S_0}}{s_n|x^n,s_0}-\f{P_{S_n|X^n,S_0}}{s_n|x^n,s^\prime_0}\vert<\varepsilon$.

\subsection{Polar Codes for Symmetric Channels}
Let $G_N = B_N F^{\otimes n}$ be Arikan's polar transform with the generator matrix for block length $N=2^n$ for $n\in\bN$.
The matrix $B_N$ is the permutation matrix called bit-reversal and is given by the recursive relation $B_N = R_N(I_2 \otimes B_{\frac{N}{2}})$ starting from $B_2=I_2$. The term $I_N$ denotes the identity matrix of size $N$ and $R_N$ denotes a permutation matrix called reverse-shuffle \cite{arikan2009channel}. 
The term $A\otimes B$ denotes the Kronecker product of $A$ and $B$ when $A,B$ are matrices, and it denotes a tensor product whenever $A,B$ are distributions. The term $A^{\otimes N}:=A\otimes A\otimes\dots\otimes A$ denotes an application of the $\otimes$ operator $N$ times. 

\par We define a polar code by the tuple $\left(\cX, \cY, W, E^W, F, G, H\right)$ that contains the channel $W$, the channels embedding $E^W$ and the core components of the \gls{sc} decoder, $F,G,H$. We define the effective bit channels by the tuple $\left(W^{(i)}_N,\cX,\cX^{i-1}\times\cY^N\right)$ for all $i\in[N]$. 
The term $E^W:\cY\to\cE$ denotes the channel embedding, where $\cE\subset\bR^d$. It is also referred in the literature by the term channel statistics, but here for our purposes, we alternatively choose to call it the channel embedding. For example, for a memoryless channel $W:=W_{Y\vert X}$, a valid choice of $E^W$, as used in the remainder of this paper, is given by the following:
\begin{equation}\label{eqn:memoryless-channel-stats}
    E^W(y) = \log\frac{\f{W}{y|1}}{\f{W}{y|0}} + \log\frac{\f{P_{X}}{1}}{\f{P_{X}}{0}},
\end{equation}
where the second term in the \gls{rhs} cancels out in the case where $P_X$ is uniform.
\par The functions $F:\cE\times\cE\to\cE, \;G:\cE\times\cE\times\cX\to\cE$ denote the check-node and bit-node operations, respectively. We denote by $H:\cE\to[0,1]$ a mapping of the embedding into a probability value, i.e. a soft decision. For the choice of $E^W$ in Equation \eqref{eqn:memoryless-channel-stats}, $F,G,H$ are given by
\begin{align}\label{eqn:sc_ops}
    &F(e_1, e_2) =  2\tanh^{-1}\left(\tanh{\frac{e_1}{2}}\tanh{\frac{e_2}{2}}\right), \nn\\
    &G(e_1, e_2, u) =  e_2 + (-1)^{u}e_1, \nn\\
    &\hspace{0.1cm}H(e_1) = \sigma(e_1),
\end{align}
where $\sigma(x) = \frac{1}{1+e^{-x}}$ is the logistic function and $e_1,e_2\in\cE, u\in\cX$.
For this choice, the hard decision rule $h:[0,1]\to\cX$ is the round function $h(l) = \mathbb{I}_{l>0.5}$, where $\bI$ is the indicator function. Applying \gls{sc} decoding on the channel outputs yields an estimate of the transmitted bits and their corresponding posterior distribution \cite{arikan2009channel}. Specifically, after observing $y^N$, \gls{sc} decoding performs the map $(y^N, f^N) \mapsto \left\{\hat{u}_i,\f{P_{U_i|U^{i-1},Y^N}}{1| \hat{u}^{i-1},y^N}\right\}_{i\in[N]}$, where $f^N$ are the frozen bits that are shared between the encoder and the decoder. i.e. $f_i\in\{0,1\}$ if $i\in[N]$ is frozen, and $f_i=0.5$\footnote{The value $0.5$ is chosen arbitrarily to indicate that the bit needs to be decoded.} if $i$ is an information bit. This mapping is denote by 
\begin{equation}\label{eqn:sc_decode}
    \left\{\hat{u}_i,\f{P_{U_i|U^{i-1},Y^N}}{1| \hat{u}^{i-1},y^N}\right\}_{i\in[N]} = \f{\mathsf{SC}_\mathsf{decode}}{y^N, f^N}.
\end{equation}
For more details on \gls{sc} decoding, the reader may refer to \cite[Section VIII]{arikan2009channel}.

\par  For the case where $P_X$ is a uniform \gls{iid} distribution, we denote by 
$$\cA = \mathsf{SC}_\mathsf{design}\left(\cD_{M,N}, k, E^W, F, G, H\right)$$
the procedure of finding the set of good channels $\cA \subset [N]$ with $|\cA|=k$ over the sample $\cD_{M,N}$ with a \gls{sc} decoder that uses $E^W$, $F$, $G$, $H$ as its elementary operations. This amounts to applying $\mathsf{SC}_\mathsf{decode}$ on the $M$ blocks in $\cD_{M,N}$. In the design phase, we assume that both $u^n$ and $y^n$ are known to the decoder, and therefore $f^N=u^N$ in the design phase. Each application of \gls{sc} decoding yields in $\left\{\f{P_{U_i|U^{i-1},Y^N}}{1| u_{j,1}^{i-1},y_{j,1}^N}\right\}_{j\in[M],i\in[N]}$ (note that the conditioning is over the true bits $u^N$). For each $i\in[N]$, we compute the empirical average 
\begin{align}\label{eq:MI_est}
  \f{\what{\sI}}{W^{(i)}_N} = 1+\sum_{j=1}^M \log \f{P_{U_i|U^{i-1},Y^N}}{u_{j,i}| u_{j,1}^{i-1},y_{j,1}^N}  
\end{align}
to estimate the \gls{mi} of the effective bit channels. Note that \eqref{eq:MI_est} follows due to the fact that $\f{\sH}{U_i\vert U^{i-1}}=1$ and the second term is an estimate of $\f{\sH}{U_i\vert U^{i-1},Y^N}$ by the law of large numbers. This estimate is used to complete the polar code design by choosing $\cA\subset[N]$ with the highest values of $\left\{\f{\what{\sI}}{W^{(i)}_N}\right\}_{i\in[N]}$.

\subsection{Neural Networks and Universal Approximations}
\par The class of shallow NNs with fixed input and output dimensions is defined as follows \cite{schafer2006recurrent}. 
\begin{definition}[NN function class]\label{def:NN_function_class}
For the ReLU activation function
$\sigma_\mathsf{R}(x) = \max(x,0)$ and $d_i,d_o \in\bN$, define the class of neural networks with $k\in\bN$ neurons as:
\begin{equation}
    \cG_\mathsf{NN}^{(d_i,k,d_o)}:=\left\{g:\bR^{d_i}\to\bR^{d_o}: g(x)=\sum_{j=1}^k\beta_j \sigma_\sR( \mathrm{W}_j x+b_j),\ x\in\bR^{d_i} \right\},\label{eq:NN_def}
\end{equation}
where $\sigma_\sR$ acts component-wise, $\beta_j\in\bR, \mathrm{W}_{j}\in \bR^{d_o \times d_i}$ and $ b_j\in\bR^{d_o}$ are the parameters of $g\in \cG_\mathsf{NN}^{(d_i,k,d_o)}$. 
Then, the class of NNs with input and output dimensions $(d_i,d_o)$ is given by
\begin{equation}
    \cG_\mathsf{NN}^{(d_i,d_o)} := \bigcup_{k\in\bN} \cG_\mathsf{NN}^{(d_i,k,d_o)},\label{eq:grnn_union}
\end{equation}
and the class of \glspl{nn} is given by $\cG_\mathsf{NN} := \bigcup_{d_i,d_o\in\bN} \cG_\mathsf{NN}^{(d_i,d_o)}$.
\end{definition}

\noindent \Glspl{nn} form a universal approximation class under mild smoothness conditions \cite{hornik1989multilayer}. The following theorem specifies
the conditions for which \glspl{nn} are universal approximators \cite[Corollary~1]{schafer2006recurrent}.
\begin{theorem}[Universal approximation of NNs]\label{lemma:uni_approx_nn_output}
Let $\sC(\cX,\cY)$ be the class continuous functions $f:\cX\to\cY$ where $\cX\subset\bR^{\mathsf{d}_i}$ is compact and $\cY\subseteq\bR^{\mathsf{d}_o}$.
Then, the class of NNs $\cG_\mathsf{NN}^{(\mathsf{d}_i,\mathsf{d}_o)}$ is dense in $\sC(\cX,\cY)$, i.e., for every $f\in\sC(\cX,\cY)$ and $\varepsilon>0$, there exist $g\in\cG_\mathsf{NN}^{(\mathsf{d}_i,\mathsf{d}_o)}$ such that
$
\|f-g\|_\infty\leq\epsilon.
$
\end{theorem}

\section{Data-Driven Polar codes for Memoryless Channels}\label{sec:data_driven_polar_code:memoryless}

\par This section focuses on designing data-driven polar codes for memoryless channels, serving as a foundational step towards addressing channels with memory. Although the algorithm presented is intended as a precursor for channels with memory—as detailed in Section \ref{sec:data_driven_polar_code:memory}—its derivation here is influenced by the eventual data-driven polar code for channels with memory. Consequently, the algorithm is not primarily aimed at enhancing existing methods for memoryless channels. Instead, its design illuminates the components essential for the data-driven polar decoder discussed in Section \ref{sec:data_driven_polar_code:memory}.

\par Let $W:=W_{Y|X}$ be a binary-input memoryless channel. Consider $\cD_{MN} \sim \left(P_{X} \otimes W_{Y|X}\right)^{\otimes MN}$ as a finite sample of its input-output pairs, with $\f{P_X}{0}=\f{P_X}{1}=0.5$.
The \gls{sc} decoding algorithm transforms the channel embedding, as detailed in \eqref{eqn:memoryless-channel-stats}, into the effective bit channels embedding using recursive formulas from \cite[Prop. 3]{arikan2009channel}. Notably, while the \gls{sc} decoder necessitates the explicit channel embedding $E^W$, the channel transition kernel remains unknown in data-driven scenarios.
To tackle this challenge, we utilize the \gls{mine} algorithm \cite{belghazi2018mine} to estimate both the channel embedding and its maximum achievable rate. Given $\cD_{MN}$, the \gls{mine} algorithm approximates $\f{\sI}{X;Y}$ using the \gls{dv} variational formula for \gls{kl} divergences. This approximation results in an estimation of the symmetric capacity (owing to the uniformity of $P_X$) as 
\begin{equation}\label{eqn:mine}
\f{\sI_\Phi}{\cD_{MN}} = \max_{\phi\in\Phi} \frac{1}{MN}\sum_{i=1}^{MN} \f{T_\phi}{x_i,y_i}-\log \frac{1}{MN}\sum_{i=1}^{MN} e^{\f{T_\phi}{x_i,\tilde{y}_i}},
\end{equation}
where $\tilde{y}^N$ represents a random shuffle of $y^N$, $T_\Phi$ is the estimated maximizer from the \gls{dv} formula, and $\Phi$ is a compact parameter space for the \gls{nn}. We represent the \gls{mine} algorithm as $T_\Phi=\f{\mathsf{MINE}}{\cD_{MN}}$.

\begin{figure}[t]
  \centering
  \begin{minipage}{0.45\textwidth}
    \centering
    \begin{algorithm}[H]
    \caption{Data-driven polar code for memoryless channels}
    \label{alg:memoryless}
    \textbf{input:} 
    Dataset $\cD_{M,N}$, \#of info. bits $k$ \\
    \textbf{output:} Clean set $\cA$
    \algrule
    \begin{algorithmic}
    \State $T_\Phi = \f{\mathsf{MINE}}{\cD_{MN}}$
    \State $E^W_\Phi = \f{T_\Phi}{1, \cdot} - \f{T_\Phi}{0, \cdot}$
    \State $\cA = \f{\mathsf{SC}_\mathsf{design}}{\cD_{M,N}, k, E^W_\Phi, F, G, \sigma}$
    \end{algorithmic}
    \end{algorithm}
  \end{minipage}
  \label{fig:algorithms-memorlyess}
\end{figure}

\par The optimal solution of the \gls{dv} formula is given by $T^\ast(x,y) = \log\frac{W(y|x)}{\frac{1}{2}W(y|0) + \frac{1}{2}W(y|1)}+c$ for $c\in\bR$ \cite{donsker1983asymptotic}. 
This connects $T^\ast$ and $E^W$ through the relation
\begin{equation}\label{eqn:mine-channel-stats-relation}
E^W(y) =  T^\ast(1,y)- T^\ast(0,y).
\end{equation}
Therefore, when the statistics of the channel are not known, the \gls{mine} algorithm's output is used as a proxy for $E^W(y)$ via Eq. \eqref{eqn:mine-channel-stats-relation}, i.e. $\f{E^W_\Phi}{y} = \f{T_\Phi}{1, y} - \f{T_\Phi}{0,y}$. 
This process is outlined in Algorithm \ref{alg:memoryless}. 
The following theorem states that Algorithm \ref{alg:memoryless} induces a consistent estimate of $\f{\sH}{U_i\vert U^{i-1}, Y^N}$ for memoryless channels.

\begin{theorem}[Successive Cancellation Decoding with \gls{mine} for Memoryless Channels]\label{thm:mine_cons}
    Let $\cD_{M,N} \sim \left(P_{X} \otimes W_{Y|X}\right)^{\otimes MN}$, where $N=2^n,\;M,n\in \bN$. Let  $u_{j,i} = (x_{j,1}^{N} G_N)_i$. Let $E_\Phi$ as defined in Equation \eqref{eqn:mine-channel-stats-relation}. Then, for every $\varepsilon>0$ there exists $p\in\bN$, compact $\Phi \subset\bR^p$ and $m\in\bN$ such that for $M>m$ and $i\in[N]$, $\bP-a.s.$
    \begin{align}\label{eqn:mine_cons}
        \left| \f{\sH_{\Phi}^M}{U_i\vert U^{i-1}, Y^N} -\f{\sH}{U_i\vert U^{i-1}, Y^N}\right| < \varepsilon,
    \end{align}
    where 
    \begin{align}\label{eqn:mine-nll}
         \f{\sH_{\Phi}^M}{U_i\vert U^{i-1}, Y^N} = -\frac{1}{M}\sum_{j=1}^M \log \f{P^{\Phi}_{U_i\vert U^{i-1}, Y^N}}{u_{j,i}\vert u_{j,1}^{i-1},y_{j,1}^{N}},
    \end{align}
    and $\f{P^{\Phi}_{U_i\vert U^{i-1}, Y^N}}{u_{j,i}\vert u_{j,1}^{i-1},y_{j,1}^{N}}$ is obtained by applying \gls{sc} decoding (using $F,G,H$ as defined in Equation \ref{eqn:sc_ops}) with inputs $\left\{\f{E^W_\Phi}{y_{j,i}}\right\}_{i=1}^N$ instead of $\left\{\f{E^W}{y_{j,i}}\right\}_{i=1}^N$.
\end{theorem}
\noindent The proof of Theorem \ref{thm:mine_cons} follows by two arguments. The first is the consistency of the \gls{mine} algorithm \cite[Theorem 2]{belghazi2018mine}. The second exploits the continuity of the \gls{sc} decoder to deduce the consistency of $\f{\sH_{\Phi}^M}{U_i\vert U^{i-1}, Y^N}$. The proof is given in Section \ref{sec:proofs}.

\begin{remark}
    The condition stated in Equation \eqref{eqn:mine_cons} is sufficient to ensure that Algorithm 1 produces a polar code such that for any $N\in\bN$ there exists $M\in\bN$ for which the block error rate is $O(2^{N^{-\beta}})$, where $\beta<\frac{1}{2}$, as the analytic polar code in \cite{arikan2009rate}.
    The law of the large numbers implies that $\lim_{M\to\infty}\f{\sH_{\Phi}^M}{U_i\vert U^{i-1}, Y^N}=\ce{P_{U_i\vert U^{i-1}, Y^N}}{P^{\Phi}_{U_i\vert U^{i-1}, Y^N}}$. This is equivalent to $$\lim_{M\to\infty}\f{\sH_{\Phi}^M}{U_i\vert U^{i-1}, Y^N}-\f{\sH}{U_i\vert U^{i-1}, Y^N}=\ckl{P_{U_i\vert U^{i-1}, Y^N}}{P^{\Phi}_{U_i\vert U^{i-1}, Y^N}}{P_{U^{i-1}, Y^N}}.$$ 
    Theorem \ref{thm:mine_cons} suggests that $\lim_{M\to\infty}\ckl{P_{U_i\vert U^{i-1}, Y^N}}{P^{\Phi}_{U_i\vert U^{i-1}, Y^N}}{P_{U^{i-1}, Y^N}}=0$ which implies that $P_{U_i\vert U^{i-1}, Y^N}$ and $P^{\Phi}_{U_i\vert U^{i-1}, Y^N}$ are equal almost everywhere. Thus, their corresponding Bhattacharyya parameters are equal and consequently have the same block error rate \cite[Theorem 3]{arikan2009channel}.
\end{remark}

\section{Data-Driven Polar codes for Channels with Memory}\label{sec:data_driven_polar_code:memory}
\par This section presents the data-driven methodology for the estimation of a neural polar decoder for channels with memory. In this case, both the channel embedding and the polar decoder needs to be estimated. The section starts with the \gls{nsc} decoder's definition, and then presents an algorithm that optimizes the \gls{nsc} decoder and the channel embedding jointly. Next, it presents neural estimation methods that allow the estimation of the channel embedding independently from the \gls{nsc} decoder. It is concluded by the theoretical guarantees of the \gls{nsc}.

\subsection{Neural Successive Cancellation Decoder}\label{sec:data_driven_polar_code:memory:nsc}
\par The following definition defines the \gls{nsc} decoder on the basis of the \gls{sc} decoder, as appears in \cite{arikan2009channel}. Specifically, it uses the structure of the \gls{sc} decoder and replaces its elementary operations by \glspl{nn}. 

\begin{definition}[Neural Successive Cancellation Decoder]\label{def:nsc}
    Let $E^W_\phi\in \cG_\mathsf{NN}$ be a channel embedding satisfying $E^W_\phi:\cY\to\cE$, $\cE\subset \bR^d$. A \gls{nsc} is defined by $F_{\theta_1}, G_{\theta_2}, H_{\theta_3}\in\cG_\mathsf{NN}$ with parameters $\theta=\{\theta_1, \theta_2, \theta_3\}$,  $\theta\in\Theta$. The \glspl{nn} satisfy:
    \begin{itemize}
        \item $F_\theta:\cE\times\cE\to\cE$ is the check-node \gls{nn}.
        \item $G_\theta:\cE\times\cE\times\cX\to\cE$ is the bit-node \gls{nn}.
        \item $H_\theta:\cE\to[0,1]$ is the soft decision \gls{nn}.
    \end{itemize}
    Application of \gls{sc} decoding, as defined in Equation \eqref{eqn:sc_decode}, with the functions $E_\phi, F_\theta, G_\theta, H_\theta$\footnote{instead of $E^W,F,G,H$ as defined in Equations \eqref{eqn:memoryless-channel-stats},\eqref{eqn:sc_ops}.} yields $P^{\phi,\theta}_{U_i\vert U^{i-1}, Y^N}$, that is an estimate of $P_{U_i\vert U^{i-1}, Y^N}$. Let
    \begin{equation}
        \f{\sH_{\phi,\theta}}{U_i\vert U^{i-1}, Y^N} = \ce{P_{U_i\vert U^{i-1}, Y^N}}{P^{\phi,\theta}_{U_i\vert U^{i-1}, Y^N}}
    \end{equation}
    be the \gls{ce} between $P_{U_i\vert U^{i-1}, Y^N}$ and $P^{\phi,\theta}_{U_i\vert U^{i-1}, Y^N}$.
\end{definition}

\par The goal of the reminder of this section is to describe how to train the parameters of the \gls{nsc}. 
Training the \gls{nsc} amounts into optimizing $\phi, \theta$ such that the symmetric capacities of the effective bit channels $\f{\sI}{W^{(i)}_N}$ are estimated. It follows that 
\begin{align}
    \f{\sI}{W^{(i)}_N} &= 1 - \f{\sH}{U_i\vert U^{i-1}, Y^N}, 
\end{align}
due to $\f{\sH}{U_i\vert U^{i-1}} = 1$ since $U_i\stackrel{iid}{\sim} \f{\mathsf{Ber}}{0.5}$.
Hence, we set the goal of estimating $P_{U_i\vert U^{i-1}, Y^N}$ as the goal needed to identify the clean effective bit channels. This implies that minimizing \gls{ce} between $P_{U_i\vert U^{i-1}, Y^N}$ and $P^{\phi,\theta}_{U_i\vert U^{i-1}, Y^N}$ is a valid objective for the optimization of $\phi,\theta$. However, in the data-driven scenario, the true distribution $P_{U_i\vert U^{i-1}, Y^N}$ is not known and therefore the \gls{ce} is computed by the negative-log-loss, which is computed exclusively by $P^{\phi,\theta}_{U_i\vert U^{i-1}, Y^N}$. 
The following definition presents the objective for the optimization of the \gls{nsc} parameters $\theta$ and the channel embedding $\phi$.

\begin{definition}[Optimization Objective of the \gls{nsc}]\label{def:nsc_opt}
    Let $\cD_{M,N} \sim P_{X^{\otimes MN}} \otimes W_{Y^{MN}\| X^{MN}}$, where $N$ is the block length and $M\in\bN$ is the number of blocks. Let $u_{j,i} = (x_{j,1}^{N} G_N)_i$. Then, for all $i\in[N]$
    \begin{equation}\label{eqn:ce_nsc_loss}
        \f{\sH_{\Phi,\Theta}^M}{U_i\vert U^{i-1}, Y^N} = \min_{\theta\in\Theta, \phi\in\Phi} \left\{ -\frac{1}{M}\sum_{j=1}^M \log \f{P^{\phi, \theta}_{U_i\vert U^{i-1}, Y^N}}{u_{j,i}\vert u_{j,1}^{i-1},y_{j,1}^{N}}\right\}
    \end{equation}
    is the objective for training the \gls{nsc}.
\end{definition}
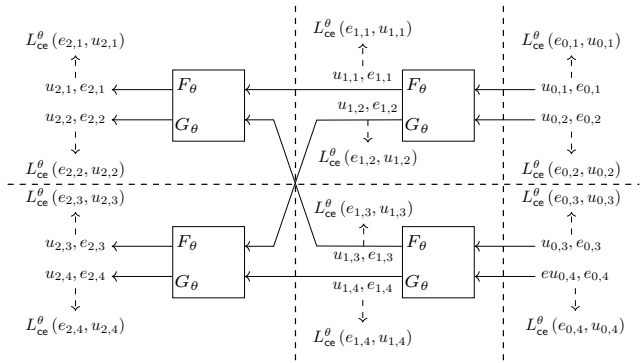
\begin{figure}[b]
    \centering
    \resizebox{0.5\textwidth}{!}{
    \tikzstyle{gn} = [rectangle,
                    node distance=1.5cm,
                    minimum width=1.5cm, 
                    minimum height=1.25cm,
                    minimum size=1.25cm,
                    text centered, 
                    draw=black,thin,
                    scale=1.0]
\tikzstyle{io} =    [node distance=0.5cm,
                    minimum height=0.25cm,
                    text centered,
                    scale=\scale]

\tikzset{XOR/.style={draw,circle,append after command={
        [shorten >=\pgflinewidth, shorten <=\pgflinewidth,]
        (\tikzlastnode.north) edge (\tikzlastnode.south)
        (\tikzlastnode.east) edge (\tikzlastnode.west)}}}

\tikzset{
    pics/nnops/.style=
    {
        code=
        {
        \node(#1)[gn] at (0,0){} ;
        \node (a) [io, at={($(#1.west)+(0.75em,1em)$)}, scale=1.5, thick] {$\scriptstyle F_\theta$};
        \node (b) [io, at={($(#1.west)+(0.75em,-1em)$)}, scale=1.5, thick] {$\scriptstyle G_\theta$};
        }
    }
} 

\def\scale{0.825}
\begin{tikzpicture}[
node distance = 3cm, 
auto,
]

\pic [] {nnops=Gn11};
\pic [right of=Gn11, xshift=1.cm] {nnops=Gn12};
\pic [below of=Gn11, xshift=0.cm, yshift=0.25cm] {nnops=Gn21};
\pic [right of=Gn21, xshift=1.cm, yshift=0.0cm] {nnops=Gn22};

\draw[->] ([yshift=0.28cm] Gn11.west) -- ++ (-3em, 0) node[anchor=east, pos=1, scale=\scale, pin={[pin edge={black, dashed, ->}, scale=\scale]90:$\celoss{u_{2,1}}{e_{2,1}}$}] {${u_{2,1}},{e_{2,1}}$};
\draw[->] ([yshift=-0.25cm]Gn11.west) -- ++ (-3em, 0) node[anchor=east, pos=1, scale=\scale, pin={[pin edge={black, dashed, ->}, scale=\scale]270:$\celoss{u_{2,2}}{e_{2,2}}$}] {${u_{2,2}},{e_{2,2}}$};
\draw[->] ([yshift=0.28cm] Gn21.west) -- ++ (-3em, 0) node[anchor=east, pos=1, scale=\scale, pin={[pin edge={black, dashed, ->}, scale=\scale]90:$\celoss{u_{2,3}}{e_{2,3}}$}] {${u_{2,3}},{e_{2,3}}$};
\draw[->] ([yshift=-0.25cm]Gn21.west) -- ++ (-3em, 0) node[anchor=east, pos=1, scale=\scale, pin={[pin edge={black, dashed, ->}, scale=\scale]270:$\celoss{u_{2,4}}{e_{2,4}}$}] {${u_{2,4}},{e_{2,4}}$};

\draw[<-] ([yshift=0.28cm]Gn11.east)  -- node[anchor=south, pos=0.75, scale=\scale, pin={[pin edge={black, dashed, ->}, scale=\scale]90:$\celoss{u_{1,1}}{e_{1,1}}$}] {${u_{1,1}},{e_{1,1}}$} ([yshift=0.28cm]Gn12.west);
\draw[<-] ([yshift=0.28cm]Gn21.east)  --++ (0.5,0.0) --++ (0.75,2.215cm) -- node[anchor=south, xshift=-0.0cm, yshift=-0.0cm, scale=\scale, pos=0.6, pin={[pin edge={black, dashed, ->}, scale=\scale]270:$\celoss{u_{1,2}}{e_{1,2}}$}] {${u_{1,2}},{e_{1,2}}$} ([yshift=-0.25cm]Gn12.west);
\draw[<-] ([yshift=-0.25cm]Gn11.east)  --++ (0.5,0.0) --++ (0.75,-2.215cm) -- node[anchor=north, xshift=-0cm, yshift=0.0cm, scale=\scale, pos=0.55, pin={[pin edge={black, dashed, ->}, scale=\scale]90:$\celoss{u_{1,3}}{e_{1,3}}$}] {${u_{1,3}},{e_{1,3}}$} ([yshift=0.28cm]Gn22.west);
\draw[<-] ([yshift=-0.25cm]Gn21.east)  -- node[anchor=north, xshift=0.0cm, yshift=0.00cm, scale=\scale, pos=0.75, pin={[pin edge={black, dashed, ->}, scale=\scale]270:$\celoss{u_{1,4}}{e_{1,4}}$}] {${u_{1,4}},{e_{1,4}}$} ([yshift=-0.25cm]Gn22.west);

\draw[<-] ([yshift=0.28cm] Gn12.east)  -- ++ (3em, 0) node[right, scale=\scale, pin={[pin edge={black, dashed, ->}, scale=\scale]90:$\celoss{u_{0,1}}{e_{0,1}}$}] {${u_{0,1}},{e_{0,1}}$};
\draw[<-] ([yshift=-0.25cm]Gn12.east)  -- ++ (3em, 0) node[right, scale=\scale, pin={[pin edge={black, dashed, ->}, scale=\scale]270:$\celoss{u_{0,2}}{e_{0,2}}$}] {${u_{0,2}},{e_{0,2}}$};
\draw[<-] ([yshift=0.28cm] Gn22.east)  -- ++ (3em, 0) node[right, scale=\scale, pin={[pin edge={black, dashed, ->}, scale=\scale]90:$\celoss{u_{0,3}}{e_{0,3}}$}] {${u_{0,3}},{e_{0,3}}$};
\draw[<-] ([yshift=-0.25cm]Gn22.east)  -- ++ (3em, 0) node[right, scale=\scale, pin={[pin edge={black, dashed, ->}, scale=\scale]270:$\celoss{u_{0,4}}{e_{0,4}}$}] {$e{u_{0,4}},{e_{0,4}}$};

\draw[-,dashed, semithick] ($(Gn11.south)-(3.5, 0.75)$) -- ($(Gn12.south)-(-3.5, 0.75)$);
\draw[-,dashed, semithick] ($(Gn11.east)+(0.88, 1.75)$) -- ($(Gn21.east)+(0.88, -1.75)$);
\draw[-,dashed, semithick] ($(Gn12.east)+(0.5, 1.75)$) -- ($(Gn22.east)+(0.5, -1.75)$);

\end{tikzpicture}
    }
    \caption{A visualization of Algorithm \ref{alg:training-nsc} for $N=4$. $L^\theta_\mathsf{ce}(e,u)$ denotes a single cross-entropy term, and the overall training loss $\mathsf{L}$ is calculated as the sum of all terms shown in the figure.}
    \label{fig:nsc-training}
\end{figure}

\par The explicit computation of $P^{\phi, \theta}_{U_i\vert U^{i-1}, Y^N}$ uses the recursive structure of the \gls{sc} decoder.
For each block $j\in[M]$ in $\cD_{M,N}$, the channel inputs and outputs $x_{j,1}^N,y_{j,1}^N$ are selected and $u_{j,1}^N$ is computed by $u_{j,1}^N=x_{j,1}^NG_N$.
For simplicity, we neglect the index of the block and focus on a single block, i.e. we simplify the notation into $x_{1}^N,y_{1}^N,u_{1}^N$. 
Let $e_{l,i}$ denotes the embedding of the $i$-th bit at the $l$-th decoding depth and $e_{l,1}^N$ denotes all the embedding at the $l$-th decoding depth. E.g. $e_{0,i}$ denotes the embedding of $X_i$ and $e_{\log_2(N), i}$ denotes the embedding of $U_i$. Accordingly, $\f{P^{\phi, \theta}_{U_i\vert U^{i-1}, Y^N}}{1|u^{i-1},y^N} = \f{H_\theta}{e_{\log_2(N),i}}$ by applying the soft-decision \gls{nn}. 

After observing $y^N$, the channel embedding are computed by $e_{0,i}=\f{E_\phi}{y_i}$. In the next step, the frozen bits $f_1^N=u_1^N$ and $e_{0,1}^N$ are used to compute the loss of the \gls{nsc}, as appears in Definition \ref{def:nsc_opt}. The loss computation is performed by a recursive function that is based on the recursion of the \gls{sc} decoder. Specifically, instead of decoding bits at each leaf of the recursion, here a loss term of entire loss is accumulated. The recursion starts with a loss accumulator initiated by $\mathsf{L}=0$. Then, the \gls{nsc} decoder starts the recursive computation of the effective bit channels, exactly as in \gls{sc} decoding, until reaching the first leaf of the recursion. Upon reaching the first leaf, the first loss term of the \gls{nsc} is accumulated into $\sL$. Precisely, a loss term $L^\theta_\mathsf{ce}(e_{\log_2(N), 1}, u_1)$ is computed via the formula of the binary \gls{ce}
\begin{equation}\label{}
    L^\theta_\mathsf{ce}(e, u) = -u\f{\log}{\f{H_\theta}{e}}-(1-u)\f{\log}{1-\f{H_\theta}{e}}.
\end{equation}
That is the binary \gls{ce} between $\f{P_{U_1\vert Y^N}}{\cdot|y^N}$ and $\f{P^{\phi, \theta}_{U_1\vert Y^N}}{\cdot|y^N}$. 
In the same manner, at each leaf of the recursion, additional loss term is accumulated to $\sL$. That is, each time reaching a leaf, $\sL$ is updated according to the following rule 
\begin{equation}
\mathsf{L} = \mathsf{L} + L^\theta_\mathsf{ce}(e_{\log_2(N),i}, u_i),\;\;i\in[N].
\end{equation}
\begin{figure}[h!]
  \centering
  \begin{minipage}{0.45\textwidth}
    \centering
    \begin{algorithm}[H]
        \caption{ $\mathsf{NSCLoss}(\mathbf{e},\mathbf{u}, \mathsf{L})$}
        \label{alg:training-nsc}
        \begin{algorithmic}
        \State $N=\f{\mathsf{dim}}{\mathbf{u}}$
        \If{$N=1$}
            \State $\mathsf{L} = \mathsf{L} +L^\theta_\mathsf{ce}(\mathbf{e},\mathbf{u})$ \Comment{Loss at a leaf}\\
            \Return $\mathsf{L},\mathbf{u}$
        \EndIf
        \State Split $\mathbf{e}$ into even and odd indices $\mathbf{e}_e, \mathbf{e}_o$
        \State $\mathbf{e}_{\mathsf{C}} = \f{F_\theta}{\mathbf{e}_e, \mathbf{e}_o}$ \Comment{Check-node}
        \State  $\mathsf{L},\mathbf{v}_1 = \f{\mathbf{\mathsf{NSCLoss}}}{\mathbf{e}_{\mathsf{C}}, \mathbf{u}_{1}^{N/2}, \mathsf{L}}$
        \vspace{0.1em}
        \State $\mathbf{e}_{\mathsf{B}} = \f{G_\theta}{\mathbf{e}_e, \mathbf{e}_o, \mathbf{v}_1}$ \Comment{Bit-node}
        \vspace{0.1em}
        \State $\mathsf{L},\mathbf{v}_2 = \f{\mathbf{\mathsf{NSCLoss}}}{\mathbf{e}_{\mathsf{B}}, \mathbf{u}_{N/2+1}^{N}, \mathsf{L}}$
        \State $\mathbf{v} = [\mathbf{v}_1, \mathbf{v}_2]$
        \State $\mathbf{v} = \mathbf{v} \left(I_{N/2} \otimes G_2\right) R_N$ \Comment{Bits in current depth}
        \State $\mathsf{L} = \mathsf{L} + \sum_{i=1}^N L^\theta_\mathsf{ce}(e_i, v_i)$ \Comment{Loss in current depth}\\
        \Return $\mathsf{L},\mathbf{v}$
        \end{algorithmic}
    \end{algorithm}
  \end{minipage}\hfill
  \begin{minipage}{0.45\textwidth}
    \centering
    \begin{algorithm}[H]
        \caption{Data-driven polar code for channels with memory}
        \label{alg:memory}
        \textbf{input:} 
        Dataset $\cD_{M,N}$, block length $n_\mathsf{t}$, \#of info. bits $k$ \\
        \textbf{output:} Clean set $\cA$
        \algrule
        \begin{algorithmic}
        \State Initiate the weights of $E^W_\phi, F_\theta, G_\theta, H_\theta$
        \State $N=2^{n_\mathsf{t}}$ \Comment{Block length in training}
        \For{$k = 1$ to $\mathsf{N_{iters}}$} 
            \State Sample $x_1^N, y_1^N\sim\cD_{M,N}$
            \State $u_1^N = x_1^NG_N$
            \State Compute $e_{0,1}^N$ by $ e_{0,i} = \f{E_\phi}{y_i}$
            \State Compute $\mathsf{L}$ by applying $\f{\mathsf{NSCLoss}}{e_{0,1}^N, u_1^N, 0}$
            \State Minimize $\mathsf{L}$ w.r.t. $E^W_\phi, F_\theta, G_\theta, H_\theta$ 
        \EndFor
        \State $\cA = \f{\mathsf{SC}_\mathsf{design}}{\cD_{M,N}, k, E^W_\phi, F_\theta, G_\theta, H_\theta}$
        \end{algorithmic}
    \end{algorithm}
  \end{minipage}
  \label{fig:algorithms-mem}
\end{figure}

\par In addition, we make the algorithm more robust by accumulating the loss incurred by bits in intermediate decoding depths of $0,1,\dots, \log{N}-1$. i.e. the loss accumulates $N (\log_2 N + 1)$ terms that correspond to all the bits in $\log_2 N + 1$ decoding depths, and for all $N$ bits per each stage. This is depicted in Figure \ref{fig:nsc-training} and Algorithm \ref{alg:training-nsc}.

\par The complete algorithm is given by the following steps, as given in Algorithm \ref{alg:memory}. First, the parameters of $E^W_\phi, F_\theta, G_\theta, H_\theta$ are initialized and the training block length is determined by $N=2^{n_\mathsf{t}}$. Then, every iteration $x_{1}^N,y_{1}^N$ are drawn from $\cD_{M,N}$ and $u_{1}^N$ is computed by $u_{1}^N=x_{1}^NG_N$.
Next, the channel embedding are computed by $e_{0,i}=\f{E_\phi}{y_i}$ for all $i\in[N]$. At this stage the loss is computed by computing $\f{\mathsf{NSCLoss}}{e_{0,1}^N, u_1^N, 0}$, as given in Algorithm \ref{alg:training-nsc}. 
The loss $\mathsf{L}$ is minimized using \gls{sgd} over the parameters $\phi,\theta$.
This procedure repeats for a predetermined amount of steps, or until the \gls{ce} stops improving. 

\subsection{Channel Embedding Estimation via Neural Estimation Methods}
\par This section claims that neural estimation methods for the estimation of the \gls{di}, as presented in \cite{tsur2022neural, tsur2022optimizing}, may be used for the estimation of the channel embedding function independently from the \gls{nsc}. The motivation for independent estimation of the channel embedding is demonstrated by memoryless channels. In this case, once the channel embedding is chosen, e.g. the \gls{llr}, the corresponding \gls{sc} decoder is compatible for all channels; the only thing that should be computed is the channel \glspl{llr}. Thus, in Section \ref{sec:data_driven_polar_code:memoryless} the channel embedding are estimated via the \gls{mine} Algorithm, and the \gls{sc} decoder is identical for all channels. In the same manner, it is desirable to find algorithms for the estimation of the channel embedding of channels with memory, such that is would be compatible with a single \gls{nsc} decoder. 

Thus, this section identifies that the output of \gls{dine} algorithm may be used to construct channel embedding for channels with memory, in the same manner as the \gls{mine} is used for memoryless channels. Section \ref{sec:memory:dine} provides a brief background on the \gls{dine} algorithm. Section \ref{sec:memory:dine_ss} shows that the \gls{dine} algorithm is a sufficient statistic of the channel outputs $Y^N$ for the estimation of $U^N$. Section \ref{sec:memory:dine_embed} shows how to extract the channel embedding from the \gls{dine} model.

\subsubsection{Estimating the Capacity of Channels with Memory}\label{sec:memory:dine}
\par Let $W:=W_{Y^{MN}\lVert X^{MN}}$ be a binary-input channel with memory and let $\cD_{MN} \sim \left(P_{X^{\otimes MN}} \otimes W_{Y^N\lVert X^N}\right)$ be a finite sample of its inputs-outputs pairs. The \gls{dine} algorithm \cite{tsur2022neural, tsur2022optimizing} estimates the \gls{di} rate from $\bX$ to $\bY$ using the following formula:
\begin{align}\label{eqn:di_estimation}
    \f{\sI_\Psi}{\bX\to \bY}&= \max_{\psi\in\Psi_{XY}} \Bigg\{\frac{1}{MN}\sum_{i=1}^{MN} \f{T_\psi}{x_i,y_i\vert x^{i-1},y^{i-1}} -\log \frac{1}{MN}\sum_{i=1}^{MN} e^{\f{T_\psi}{x_i,z_i\vert x^{i-1},y^{i-1}}}\Bigg\} \nonumber\\
    &-\max_{\psi\in\Psi_{Y}} \Bigg\{\frac{1}{MN}\sum_{i=1}^{MN} \f{T_\psi}{y_i\vert y^{i-1}}-\log \frac{1}{MN}\sum_{i=1}^{MN} e^{\f{T_\psi}{z_i\vert y^{i-1}}}\Bigg\},
\end{align}
where $T_\psi\in\rnn$, the space of \glspl{rnn} whose parameter space is $\Psi$. The \glspl{rv} $Z^N$ are auxiliary \gls{iid} \glspl{rv} distributed on $\cY^N$ and independent of $X^N,Y^N$, that are used for the estimation of the \gls{di} as presented in \cite{tsur2022neural}. The estimated maximizers of the first and second terms are denoted by $T_{\Psi_{XY}}$ and $T_{\Psi_{Y}}$, respectively. 

\subsubsection{DINE gives Sufficient Statistics of the Channel Outputs}\label{sec:memory:dine_ss}
\par The optimal maximizers of the $i$-th argument in first term in Equation \eqref{eqn:di_estimation} is given by $T_i^\ast= \log\frac{P_{Y_i\vert Y^{i-1},X^{i}}}{P_Z} +c$ for $c\in\bR$ and $i\in\bN$. 
For fixed $y^N$, we define a new \gls{rv} $T^\ast_{y^i}: \cX^i\to\bR$ by
\begin{align}\label{eqn:suffiecient_stat_def}
    \f{T^\ast_{y^i}}{x^i} = \log\frac{\f{P_{Y_i\vert Y^{i-1},X^{i}}}{y_i \vert y^{i-1},x^i}}{\f{P_Z}{y_i}}.
\end{align}
The following theorem states that $T^N \triangleq \left\{T^\ast_{Y^i}\right\}_{i=1}^N$ is a sufficient statistic of $Y^N$ for the estimation $U^N$.
 \begin{theorem}\label{thm:channel_sufficient}
 Let $X^N,Y^N \sim P_{X^N} \otimes W_{Y^N\Vert X^N}$ and $P_Z$ such that $P_Y\ll P_Z$. Then $T^N$, as defined in Equation \eqref{eqn:suffiecient_stat_def}, satisfies
 \begin{align}
     &U^N - Y^N - T^N, \\ 
     &U^N - T^N - Y^N.
 \end{align}
 \end{theorem}  

\noindent The proof of Theorem \ref{thm:channel_sufficient} is given in Section \ref{sec:proofs}. The main steps of the proof are to express $P_{X^N,Y^N}$ in terms of $T^N$ and use the well-known Fisher-Neyman Factorization theorem.

\subsubsection{Obtaining the parametric channel embedding from the DINE model}\label{sec:memory:dine_embed}
\par Theorem \ref{thm:channel_sufficient} suggests that \gls{di} estimation is an appropriate objective for the construction of the channel embedding of $Y^N$ needed for the \gls{nsc} decoder. 
However, the evaluation of $T^N$ for all $x^N\in\cX^N$ involves an exponential number of computations. To overcome this, recall that according to Equation \eqref{eqn:di_estimation}, $T_{\Psi_{XY}}$ is approximated by a \gls{rnn} that contains a sequence of layers. Therefore, we design $T_{\psi_{XY}}$ to process $x^N, e^N$ instead of processing $x^N,y^N$. Specifically, we denote this construction by 
$\f{T_{\psi_{XY}, \phi}}{x^i,y^i} = \f{\widetilde{T}_{\Psi_{XY}}}{x^i,e^i}$,
where 
\begin{equation}\label{eqn:nsc_embedding}
e_i=\f{E^W_\phi}{y_i},\; E^W_\phi:\cY\to \bR^d    
\end{equation} 
is an embedding of $y_i$. 

\par With this parameterization, after applying \cite[Algorithm 1]{tsur2022optimizing} we obtain $T_{\Psi_{XY}, \Phi}$, that contains $E^W_\Phi$ as its intermediate layer. Since $T_{\Psi_{XY}, \Phi}$ is composed of sequential layers, any intermediate layer of $T_{\psi_{XY}, \phi}$ must preserve the information that flows to its outputs. Therefore, we choose $E^W_\Phi$ to be the channel embedding required for the \gls{nsc} decoder. For this choice, the parameters of the channel embedding are fixed and the \gls{nsc} can be optimized without the optimization over $\Phi$. Specifically, the minimization in Definition \ref{def:nsc_opt} is perform exclusively over $\Theta$. 

\subsection{Consistency}
The next theorem shows the consistency of the \gls{nsc} for channels with memory. It demonstrates that, for \glspl{fsc}, Algorithm \ref{alg:memory} yields a consistent estimator of the \gls{sc} polar decoder.

\begin{theorem}[Successive Cancellation Decoding of Time Invariant Channels (restated)]\label{thm:nsc}
    Let $\bX,\bY$ be the inputs and outputs of an indecomposable \gls{fsc} as given in Section \ref{sec:pre:fsc}. Let $\cD_{M,N} \sim P_{X^{MN}} \otimes W_{Y^{MN}\| X^{MN}}$, where $N=2^n,\;M,n\in \bN$. Let  $u_{j,i} = (x_{j,1}^{N} G_N)_i$. Then, for every $\varepsilon>0$ there exists $p\in\bN$, compact $\Phi,\Theta \in\bR^p$ and $m\in\bN$ such that for $M>m$ and $i\in[N]$, $\bP-a.s.$
    \begin{align}
        \left| \f{\sH_{\Phi,\Theta}^M}{U_i\vert U^{i-1}, Y^N} -\f{\sH}{U_i\vert U^{i-1}, Y^N}\right| < \varepsilon.
    \end{align}
\end{theorem}
\noindent Theorem \ref{thm:nsc} concludes that there exists \glspl{nn} that approximate the \gls{sc} elementary operations with an arbitrary precision. It also indicates that these operations do not depend on the specific block, or the specific symbol location inside the block. i.e. the same \glspl{nn}, $E_\phi,F_\theta,G_\theta, H_\theta$, may be used for all decoding stages and for all bits inside each decoding stage.
\par The proof starts with identifying that the structure of $\f{P^{\phi, \theta}_{U_i\vert U^{i-1}, Y^N}}{1|u^{i-1},y^N}$ is induced by the structure of the \gls{sc} decoder and that it contains $4$ unique operations, $E_\phi^W, F_\theta,G_\theta, H_\theta$, that operate on channel embedding in $\bR^d$. It continues with an approximation step, in which $\f{P_{U_i\vert U^{i-1}, Y^N}}{1|u^{i-1},y^N}$ is parameterized by \gls{nn} via the universal approximation theorem of \glspl{nn} \cite{hornik1989multilayer}. Then, an estimation step follows, in which expected values are estimated by empirical means via the uniform law of large numbers for stationary and ergodic processes \cite{pevskir1996uniform}. The full proof is given in Section \ref{sec:proofs}.
\subsection{Computational Complexity}
\par The following theorem examines the computational complexity of the \gls{nsc} decoder for the case where $E_\phi, F_\theta, G_\theta, H_\theta$ are \glspl{nn} with $k$ hidden units and the embedding space satisfies $\cE\subset \bR^d$. 
\begin{theorem}[Computational Complexity of the \gls{nsc}]\label{thm:nsc_complexity}
    Let $E_\phi\in\cG_\mathsf{NN}^{(1,k,d)},F_\theta\in\cG_\mathsf{NN}^{(2d,k,d)},G_\theta\in\cG_\mathsf{NN}^{(2d+1,k,d)}, H_\theta\in\cG_\mathsf{NN}^{(d,k,1)}$. Then, the computational complexity of \gls{nsc} decoding is $\f{O}{kd N\log_2 N}$.
\end{theorem}
\begin{proof}
    According to \cite[Section VIII]{arikan2009channel}, the recursive formulas of the \gls{sc} decoder have complexity of $O(N\log N)$. In \cite[Section VIII]{arikan2009channel} the decoding operations have complexity of $O(1)$ and therefore do not affect the overall complexity. Here, we consider decoding operations that are given by \glspl{nn} with input dimension at most $2d+1$, $k$ hidden units and output dimension of at most $d$. The complexity of such \gls{nn} is $O(kd)$ that yields in the overall complexity of the \gls{nsc} to be $O(kdN\log N)$
\end{proof}
\par The only difference between Theorem \ref{thm:nsc_complexity} and \cite[Theorem 5]{arikan2009channel} is that the \gls{nn} computation complexity is given explicitly by $kd$ (even though it could have been neglected as it does not depend on $N$ or the channel's state space). The goal of Theorem \ref{thm:nsc_complexity} is to compare the \gls{nsc} decoder with \gls{sct} decoder. Recall that the computational complexity of the \gls{sct} decoder is $\f{O}{|\cS|^3 N\log N}$. This sets a main advantage of the \gls{nsc} decoder -- its computational complexity does not grow with the memory size of the channel. 

\section{Extension to Asymmetric Input Distributions via Honda-Yamamoto Scheme}\label{sec:hy}
\par This section describes how to extend the data-driven polar decoder to the case where the input distribution is not necessarily symmetric. It starts with a brief description of the Honda-Yamamoto scheme \cite{honda2013polar}. Then, it extends our methods to accommodate asymmetric input distributions by incorporating this scheme. In the case where the channel is memoryless, the extension is straight-forward and it is described in Section \ref{sec:hy:memoryless}. For the case where the channel has memory, the extension entails applying the \gls{nsc} twice, as done in the Honda-Yamamoto scheme. This is described in Section \ref{sec:hy:memory}.
\subsection{Honda-Yamamoto Scheme for Asymmetric Channels}
\par The Honda-Yamamoto scheme \cite{honda2013polar} generalizes polar coding for asymmetric input distributions. Here, the polar decoder is applied twice; first, before observing the channel outputs and second, after observing the channel outputs. An equivalent interpretation is that the first application of \gls{sc} decoding is done on a different channel whose outputs are independent of its inputs. Indeed, in this case, as given in Equation \eqref{eqn:memoryless-channel-stats}, the first term of the \gls{rhs} cancels out, and it follows that the channel embedding are constant for all $y\in\cY$. Thus, for the first application of \gls{sc} decoding, we denote the constant \emph{input embedding} by $E^X$ (rather than $E^W$).
The second application of \gls{sc} decoding follows the same procedure as in the case for symmetric channels.

\par Accordingly, a polar decoder with non symmetric input distribution is defined by the tuple $\left(\cX, \cY, W, E^X, E^W, F, G, H\right)$. Here, we add the input embedding $E^X$ to the definition, where $E_X(y)$ is constant for all $y\in\cY$. An important observation is that the functions $F,G,H$ are independent of the channel, i.e. both application of \gls{sc} decoding (before and after observing the channel outputs) share the same functions $F,G,H$. 

\par Given a finite sample $\cD_{M,N}$, we denote by $$\cA = \mathsf{SC}_\mathsf{design-HY}\left(\cD_{M,N}, k, E^X, E^W, F, G, H\right)$$ the procedure of finding the set of good channels $\cA \subset [N]$ with $|\cA|=k$ over the sample $\cD_{M,N}$ with a \gls{sc} decoder that uses $E^X$, $E^W$, $F$, $G$, $H$ as its elementary operations. 
Specifically, for each block $j\in[M]$, \gls{sc} decoding is applied twice. First, $E^X$ is used to compute the channel embedding; it yields in the computation of $\left\{\f{P_{U_i|U^{i-1}}}{1| u_{j,1}^{i-1}}\right\}_{j\in[M],i\in[N]}$. 
Next, $E^W$ is used to compute the channel embedding; it yields in the computation of $\left\{\f{P_{U_i|U^{i-1},Y^N}}{1| u_{j,1}^{i-1},y_{j,1}^N}\right\}_{j\in[M],i\in[N]}$.
For each $i\in[N]$, we compute the empirical average 
$$\f{\what{\sI}}{W^{(i)}_N} = -\frac{1}{M}\sum_{j=1}^M \log \f{P_{U_i|U^{i-1}}}{u_{j,i}| u_{j,1}^{i-1}}+\frac{1}{M}\sum_{j=1}^M \log \f{P_{U_i|U^{i-1},Y^N}}{u_{j,i}| u_{j,1}^{i-1},y_{j,1}^N}$$
to estimate the \gls{mi} of the effective bit channels. This estimate is used to complete the polar code design by choosing $\cA\subset[N]$ with the highest values of $\left\{\f{\what{\sI}}{W^{(i)}_N}\right\}_{i\in[N]}$.

\subsection{Memoryless Channels}\label{sec:hy:memoryless}
The case of asymmetric input distributions is similar to the case of polar code design for uniform input distributions. We deal with non-uniform input distributions by applying the Honda-Yamamoto scheme \cite{honda2013polar}. 
Consider $W:=W_{Y|X}$ be a binary-input memoryless channel and $P_X$ an \gls{iid} and non-uniform input distribution\footnote{In the case of memoryless channels it is sufficient to consider an \gls{iid} input distribution since it achieves capacity \cite{cover2006elements}.}. Accordingly, given $\cD_{MN}$, also here, once $T_\Phi$ is estimated via the \gls{mine} algorithm, it is used as a proxy of the channel embedding $E^W$ by the formula 
\begin{equation}
    \f{E_\Phi^W}{y} = \f{T_\Phi}{1,y} - \f{T_\Phi}{0,y} + \log\frac{\f{P_X}{1}}{\f{P_X}{0}}
\end{equation} 
and the Honda-Yamamoto scheme is applied ``as-is''. Specifically, for this case, Algorithm \ref{alg:memoryless} is applied with the only exception that the polar code design $\mathsf{SC}_\mathsf{design}$ is replaced by $\mathsf{SC}_\mathsf{design-HY}$.
\begin{algorithm}[t]
    \caption{Data-driven polar code design for channels with memory and non-\gls{iid} input distribution}
    \label{alg:memory-hy}
    \textbf{input:} 
    Input distribution $P_{X^N}$, Channel $W_{Y^N\lVert X^N}$, block length $n_\mathsf{t}$, \#of info. bits $k$ \\
    \textbf{output:} Clean set $\cA$
    \algrule
    \begin{algorithmic}
    \State Initiate the weights of $E^X_\phi,E^Y_\phi$ and $G_\theta, F_\theta,H_\theta$
    \State $N=2^{n_\mathsf{t}}$
    \For{$k = 1$ to $\mathsf{N_{iters}}$} 
        \State Sample $x^N, y^N \sim P_{X^N}\otimes W_{Y^N\Vert X^N}$
        \State $u^N = x^NG_N$
        \State Duplicate $e_X$ to $e_X^N$ 
        \State Compute $e_Y^N$ by $ e_i = \f{E}{y_i}$
        \State Compute $\sL_X$ by applying $\f{\mathsf{NSCTrain}}{e_X^N, u^N, 0}$
        \State Compute $\sL_Y$ by applying $\f{\mathsf{NSCTrain}}{e_Y^N, u^N, 0}$
        \State Minimize $\sL_X + \sL_Y$ w.r.t. $\phi,\theta$. 
    \EndFor
    \State $\cA = \f{\mathsf{SC}_\mathsf{design}}{\cD_{M,N}, k, E^X_\phi,E^W_\phi, G_\theta, F_\theta, H_\theta}$
    \end{algorithmic}
\end{algorithm}
\subsection{Channels with Memory} \label{sec:hy:memory}
\par This section considers two issues. The first is the choice of an input distribution. This is addressed by employing algorithms for capacity estimation \cite{tsur2022neural, tsur2022optimizing}. The second issue addresses the construction of a \gls{nsc} decoder that is tailored for stationary input distributions.

\par For the choice of the input distribution, we employ recent method for the optimization of the \gls{dine} as presented in \cite{tsur2022optimizing}. Therein, the authors provide an \gls{rl} algorithm that uses \gls{dine} to estimate capacity achieving input distributions. 
The input distribution is approximated with an \gls{rnn} with parameter space denoted by $\Pi$. Let $P_X^\pi$ be the estimated capacity achieving input distribution. Thus, by application of \cite[Algorithm 1]{tsur2022optimizing}, we obtain a model of $P_X^\pi$ from which we are able to sample observations of the channel inputs. 



\par Extension of Algorithm \ref{alg:memory} to $P_{X^N}$ (that is not uniform and \gls{iid}) involves introducing additional parameters, that we denote by $\phi_2\in\Phi$. Accordingly, we denote the set of the channel embedding by $\phi=\{\phi_1, \phi_2\}$, where $\phi_1$ denotes the parameters of $E^X$ and $\phi_2$ are the parameters of $E^W$. We define $E^X_\phi:\cY\to\bR^d$ as a constant \gls{rv} that satisfies $\f{E^X_\phi}{y}=e_X\in\bR^d$ for all $y\in\cY$. Accordingly, the \gls{nsc} in this case is defined by $E^X_\phi,E^W_\phi,F_\theta,G_\theta,F_\theta$. Thus, Algorithm \ref{alg:memory} needs to be updated in order to optimize $E^X_\phi$ as well. This is addressed by first applying the \gls{nsc} with inputs $e_X^N$ to compute $P^{\phi, \theta}_{U_i\vert U^{i-1}}$, where $e_X^N\in\bR^{d\times N}$ is a matrix whose columns are duplicates of $e_X$. Second, the \gls{nsc} is applied with $e^N_Y$ to compute $P^{\phi,\theta}_{U_i\vert U^{i-1}, Y^N}$, where $e_Y^N\in\bR^{d\times N}$ is a matrix whose $i$-th column is $\f{E^W_\phi}{y_i}$.
\begin{figure*}[b]
     \centering
     \hspace{-1.05cm}
     \begin{subfigure}[]{0.48\textwidth}
        \centering
\begin{tikzpicture}

\definecolor{chocolate2267451}{RGB}{226,74,51}
\definecolor{dimgray85}{RGB}{85,85,85}
\definecolor{gainsboro229}{RGB}{229,229,229}

\begin{axis}[
width=2.45in,
height=2.3in,
log basis y={10},
x grid style={gainsboro229},
xlabel=\textcolor{dimgray85}{\(\displaystyle n\)},
xmajorgrids,
xmin=2.6, xmax=10.5,
xtick style={color=dimgray85},
y grid style={gainsboro229},
ylabel=\textcolor{dimgray85}{Bit error rate},
ymajorgrids,
ymin=3e-06, ymax=0.168600008718341,
ymode=log,
ytick style={color=dimgray85},
legend style={at={(0.05,0.05)},anchor=south west, nodes={scale=0.85, transform shape} }
]
\addplot [semithick, violet, mark=*, mark size=1, mark options={solid}]
table {%
3 0.0405924
4 0.0224246
5 0.0162571884984026
6 0.00384976114649681
7 0.00248592251630226
8 0.000826343283582089
9 9.30725682938735e-05
10 1.11070110701107e-05
};
\addlegendentry{$E^W$}

\addplot [semithick, chocolate2267451, mark=*, mark size=1, mark options={solid}]
table {%
3 0.0410296
4 0.0187896
5 0.0170403354632588
6 0.00428941082802548
7 0.00289547372458765
8 0.000913059701492537
9 9.60740740740741e-05
10 1.28518518518519e-05
};
\addlegendentry{$E^W_\Phi$}
\end{axis}

\end{tikzpicture}         
        \label{fig:ber_bsc}
     \end{subfigure}
     \hspace{0.025cm}
     \begin{subfigure}[]{0.48\textwidth}
     \vspace{.25cm}
        \centering
\begin{tikzpicture}

\definecolor{chocolate2267451}{RGB}{226,74,51}
\definecolor{dimgray85}{RGB}{85,85,85}
\definecolor{gainsboro229}{RGB}{229,229,229}

\begin{axis}[
width=2.45in,
height=2.3in,
log basis y={10},
x grid style={gainsboro229},
xlabel=\textcolor{dimgray85}{\(\displaystyle n\)},
xmajorgrids,
xmin=2.6, xmax=10.5,
xtick style={color=dimgray85},
y grid style={gainsboro229},
ylabel=\textcolor{dimgray85}{Bit error rate},
ymajorgrids,
ymin=3e-06, ymax=0.168600008718341,
ymode=log,
ytick style={color=dimgray85},
legend style={at={(0.05,0.05)},anchor=south west, nodes={scale=0.85, transform shape} }
]
\addplot [semithick, violet, mark=*, mark size=1, mark options={solid}]
table {%
3	0.041
4	0.02
5	0.011625
6	0.00784375
7	0.0039375
8	0.001269737
9	0.000271896
10	3.99727E-05
};
\addlegendentry{$E^W$}

\addplot [semithick, chocolate2267451, mark=*, mark size=1, mark options={solid}]
table {%
3 0.042
4 0.024
5 0.01175199900049975
6 0.00812231384307846
7 0.004150862068965517
8 0.0013808783108445777
9 0.0002945002498750625
10 0.00004231930909545227
};
\addlegendentry{$E^W_\phi$}
\end{axis}

\end{tikzpicture}         
        \label{fig:ber_awgn}
     \end{subfigure}
     \caption{The left figure compares the \gls{ber} incurred by  Algorithm \ref{alg:memoryless} and the \gls{sc} decoder on a BSC with parameter ${0.1}$. The right figure compares the \gls{ber} incurred by Algorithm \ref{alg:memoryless} and the \gls{sc} decoder on a AWGN channel. The curves labeled by $E^W, E^W_\Phi$ correspond to the analytic and estimated channel embedding, respectively.}
     \label{fig:memoryless-symmetric}
\end{figure*}
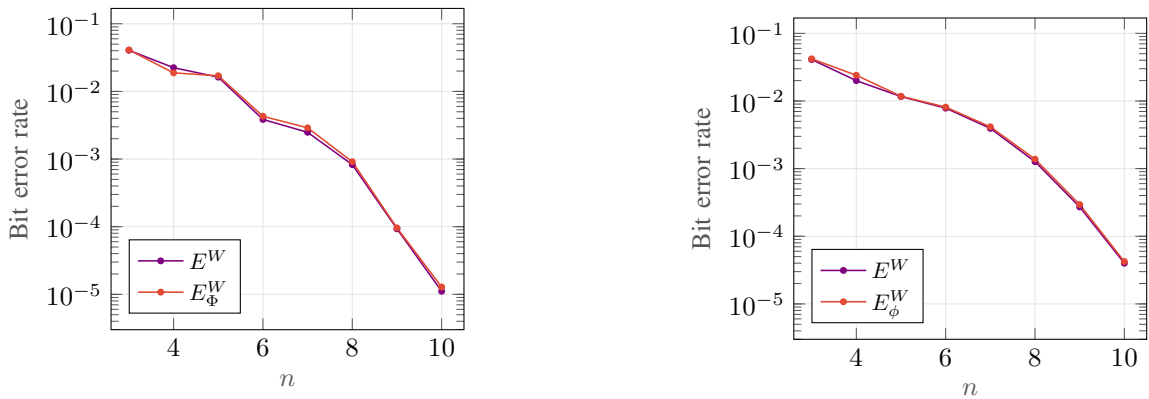
\par The training procedure admits the following steps. First, the channel inputs and outputs are sampled by $x^N,y^N \sim P^\pi_{X^N}\otimes W_{Y^N\Vert X^N}$. Then, the values of $u^N = x^N G_N$ are computed, and form the labels of the algorithm. Next, the channel statistics $e_Y^N$ are computed and the input statistics are duplicated to obtain $e_X^N$. The next step is to apply the \gls{nsc}-Train procedure twice, i.e.
\begin{align}
    \sL_X &= \mathsf{NSCTrain}(e_X^N, u^N, 0) \\
    \sL_Y &= \mathsf{NSCTrain}(e_Y^N, u^N, 0),
\end{align}
which are minimized via \gls{sgd}. This procedure is depicted in Algorithm \ref{alg:memory-hy}.

\begin{figure*}[t]
     \centering
     \hspace{-1.05cm}
     \begin{subfigure}[]{0.48\textwidth}
        \centering
\begin{tikzpicture}

\definecolor{chocolate2267451}{RGB}{226,74,51}
\definecolor{dimgray85}{RGB}{85,85,85}
\definecolor{gainsboro229}{RGB}{229,229,229}

\begin{axis}[
width=2.45in,
height=2.3in,
log basis y={10},
x grid style={gainsboro229},
xlabel=\textcolor{dimgray85}{\(\displaystyle n\)},
xmajorgrids,
xmin=2.6, xmax=10.5,
xtick style={color=dimgray85},
y grid style={gainsboro229},
ylabel=\textcolor{dimgray85}{Bit error rate},
ymajorgrids,
ymin=3e-06, ymax=0.168600008718341,
ymode=log,
ytick style={color=dimgray85},
legend style={at={(0.05,0.05)},anchor=south west, nodes={scale=0.85, transform shape} }
]
\addplot [semithick, violet, mark=*, mark size=1, mark options={solid}]
table {%
3	0.05
4	0.04
5	0.03175
6	0.01875
7	0.01346875
8	0.0088125
9	0.00411875
10	0.001868219
};
\addlegendentry{$P_X(1)=\frac{9}{16}, E^W_\Phi$}

\addplot [semithick, chocolate2267451, mark=*, mark size=1, mark options={solid}]
table {%
3  0.05
4  0.0325
5  0.016
6  0.011375
7  0.00628125
8  0.001690341
9  0.000422341
10 0.000119641
};
\addlegendentry{$P_X(1)=\frac{1}{2}, E^W_\Phi$}
\end{axis}

\end{tikzpicture}         
        \label{fig:ber_bec-hy}
     \end{subfigure}
     \hspace{0.025cm}
     \begin{subfigure}[]{0.48\textwidth}
     \vspace{.25cm}
        \centering
\begin{tikzpicture}

\definecolor{chocolate2267451}{RGB}{226,74,51}
\definecolor{dimgray85}{RGB}{85,85,85}
\definecolor{gainsboro229}{RGB}{229,229,229}

\begin{axis}[
width=2.45in,
height=2.3in,
log basis y={10},
x grid style={gainsboro229},
xlabel=\textcolor{dimgray85}{\(\displaystyle n\)},
xmajorgrids,
xmin=2.6, xmax=10.5,
xtick style={color=dimgray85},
y grid style={gainsboro229},
ylabel=\textcolor{dimgray85}{Bit error rate},
ymajorgrids,
ymin=3e-06, ymax=0.168600008718341,
ymode=log,
ytick style={color=dimgray85},
legend style={at={(0.05,0.05)},anchor=south west, nodes={scale=0.85, transform shape} }
]

\addplot [semithick, violet, mark=*, mark size=1, mark options={solid}]
table {%
3	0.05
4	0.04
5	0.03175
6	0.022875
7	0.01246875
8	0.0088125
9	0.00411875
10	0.001868219
};
\addlegendentry{$P_X(1)=\frac{9}{16}, E^W$}

\addplot [semithick, chocolate2267451, mark=*, mark size=1, mark options={solid}]
table {%
3  0.045
4  0.032
5  0.019
6  0.008875
7  0.0058125
8  0.001683036
9  0.000445313
10 7.27371E-05
};
\addlegendentry{$P_X(1)=\frac{1}{2}, E^W$}
\end{axis}

\end{tikzpicture}         
        \label{fig:ber_bec-hy-analytic}
     \end{subfigure}
     \caption{The figure compares the \glspl{ber} incurred on the asymmetric BEC; both compare the results for $P_X(1)=0.5$ and $P_X(1)=\frac{9}{16}$ (capacity achieving input distribution). The left figure compares the \glspl{ber} incurred by applying the memoryless algorithm (Algorithm \ref{alg:memoryless}) for both choices of $P_X$. The right figure is the ground truth of the left figure by using $E_W$ instead of $E^\Phi_W$.}
     \label{fig:memoryless-non-symmetric}
\end{figure*}
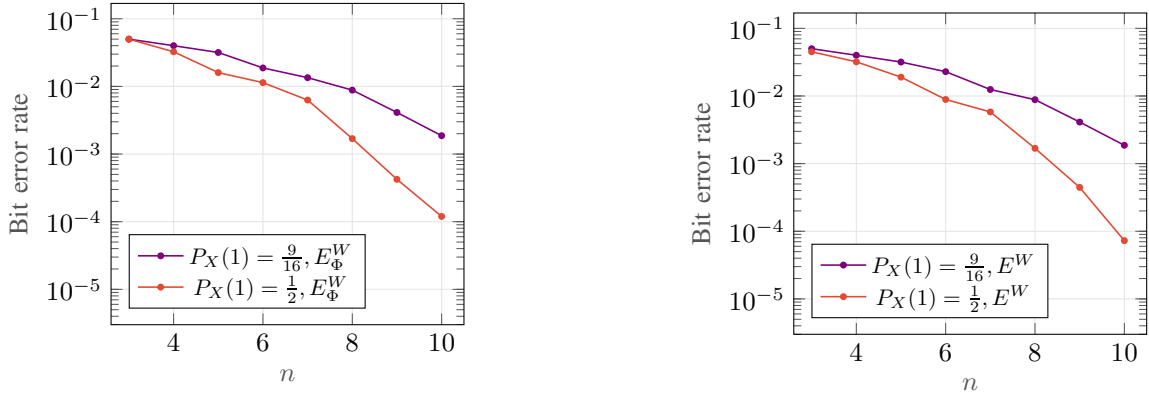

\section{Experiments}\label{sec:exp}
\par This section presents experiments on memoryless channels and channels with memory. The experiments demonstrate the performance of the proposed algorithms for symmetric and non-symmetric input distributions. Section \ref{sec:exp:memoryless} considers memoryless channels, and Section \ref{sec:exp:memory} considers channels with memory. Each section includes the results for both cases where the input distribution is uniform or not. All the polar codes in this sections, unless specified otherwise, are designed with rate $R=0.25$.

\begin{figure*}[b]
\def\x{0.05cm}
     \centering
     \hspace{1.0cm}
     \begin{subfigure}[]{0.2\textwidth}
        \centering
        \hspace{-1.4cm}
\begin{tikzpicture}

\definecolor{chocolate2267451}{RGB}{226,74,51}
\definecolor{dimgray85}{RGB}{85,85,85}
\definecolor{gainsboro229}{RGB}{229,229,229}

\begin{axis}[
width=1.96in,
height=1.84in,
log basis y={10},
x grid style={gainsboro229},
xlabel=\textcolor{dimgray85}{\(\displaystyle n\)},
xmajorgrids,
xmin=2.6, xmax=10.5,
xtick style={color=dimgray85},
y grid style={gainsboro229},
ylabel=\textcolor{dimgray85}{Bit error rate},
ymajorgrids,
ymin=3e-06, ymax=0.168600008718341,
ymode=log,
ytick style={color=dimgray85},
legend style={at={(0.05,0.05)},anchor=south west, nodes={scale=0.85, transform shape} }
]
\addplot [semithick, violet, mark=*, mark size=1, mark options={solid}]
table {%
3	0.041
4	0.02
5	0.011625
6	0.00784375
7	0.0039375
8	0.001269737
9	0.000271896
10	3.99727E-05
};
\addlegendentry{\scriptsize SC}

\addplot [semithick, chocolate2267451, mark=*, mark size=1, mark options={solid}]
table {%
3  0.042
4  0.022
5  0.011751999
6  0.008122314
7  0.004150862
8  0.001380878
9  0.0003245
10 4.23193E-05
};
\addlegendentry{\scriptsize NSC}
\end{axis}

\end{tikzpicture}         
        \label{fig:symmetric-awgn}
        \vspace{-0.45cm}
        \caption{AWGN Channel}
     \end{subfigure}
     \begin{subfigure}[]{0.2\textwidth}
        \centering
\begin{tikzpicture}

\definecolor{chocolate2267451}{RGB}{226,74,51}
\definecolor{dimgray85}{RGB}{85,85,85}
\definecolor{gainsboro229}{RGB}{229,229,229}

\begin{axis}[
width=1.96in,
height=1.84in,
log basis y={10},
x grid style={gainsboro229},
xlabel=\textcolor{dimgray85}{\(\displaystyle n\)},
xmajorgrids,
xmin=2.6, xmax=10.5,
xtick style={color=dimgray85},
y grid style={gainsboro229},
ymajorgrids,
ymin=3e-06, ymax=0.168600008718341,
ymode=log,
ymajorticks=false,
legend style={at={(0.05,0.05)},anchor=south west, nodes={scale=0.85, transform shape} }
]
\addplot [semithick, violet, mark=*, mark size=1, mark options={solid}]
table {%
3	0.034
4	0.017
5	0.01225
6	0.00778125
7	0.0018375
8	0.00050038
9	0.00011935
10	2.98125E-05
};
\addlegendentry{\scriptsize SCT}

\addplot [semithick, chocolate2267451, mark=*, mark size=1, mark options={solid}]
table {%
3  0.037
4  0.0175
5  0.01375
6  0.0086875
7  0.001975
8  0.000501838
9  0.000129828
10 2.90117E-05
};
\addlegendentry{\scriptsize NSC}
\end{axis}

\end{tikzpicture}         
        \label{fig:symmetric-ising}
        \vspace{-0.5cm}
        \hspace{2cm}
        \caption{Ising Channel}
     \end{subfigure}
     \hspace{\x}
     \begin{subfigure}[]{0.2\textwidth}
        \centering
\begin{tikzpicture}

\definecolor{chocolate2267451}{RGB}{226,74,51}
\definecolor{dimgray85}{RGB}{85,85,85}
\definecolor{gainsboro229}{RGB}{229,229,229}

\begin{axis}[
width=1.96in,
height=1.84in,
log basis y={10},
x grid style={gainsboro229},
xlabel=\textcolor{dimgray85}{\(\displaystyle n\)},
xmajorgrids,
xmin=2.6, xmax=10.5,
xtick style={color=dimgray85},
y grid style={gainsboro229},
ymajorgrids,
ymin=3e-06, ymax=0.168600008718341,
ymode=log,
ymajorticks=false,
legend style={at={(0.05,0.05)},anchor=south west, nodes={scale=0.85, transform shape} }
]
\addplot [semithick, violet, mark=*, mark size=1, mark options={solid}]
table{
3  0.023
4  0.015
5  0.0055
6  0.00437
7  0.00109
8  0.00046
9  0.0001053
10 0.0000294
};
\addlegendentry{\scriptsize SCT}
\addplot [semithick, chocolate2267451, mark=*, mark size=1, mark options={solid}]
table{
3  0.025
4  0.016
5  0.0057
6  0.0045
7  0.0011
8  0.00048
9  0.000113
10 0.0000315
};
\addlegendentry{\scriptsize NSC}
\end{axis}

\end{tikzpicture}         
        \label{fig:symmetric-interference}
        \vspace{-0.5cm}
        \caption{ISI Channel}
     \end{subfigure}
    \hspace{\x}
     \begin{subfigure}[]{0.2\textwidth}
        \centering
\begin{tikzpicture}

\definecolor{chocolate2267451}{RGB}{226,74,51}
\definecolor{dimgray85}{RGB}{85,85,85}
\definecolor{gainsboro229}{RGB}{229,229,229}

\begin{axis}[
width=1.96in,
height=1.84in,
log basis y={10},
x grid style={gainsboro229},
xlabel=\textcolor{dimgray85}{\(\displaystyle n\)},
xmajorgrids,
xmin=2.6, xmax=10.5,
xtick style={color=dimgray85},
ymajorgrids,
ymin=3e-06, ymax=0.168600008718341,
ymode=log,
ymajorticks=false,
legend style={at={(0.05,0.05)},anchor=south west, nodes={scale=0.85, transform shape} }
]

\addplot [semithick, chocolate2267451, mark=*, mark size=1, mark options={solid}]
table {%
3  0.024
4  0.017
5  0.008104166666666668
6  0.007875
7  0.00328125
8  0.001255055147058823
9  0.000468674150485436
10 0.000059498046875
};
\addlegendentry{\scriptsize NSC}
\end{axis}

\end{tikzpicture}         
        \label{fig:symmetric-maagn}
        \vspace{-0.5cm}
        \caption{MA-AGN Channel}
     \end{subfigure}
     \caption{The figures compares the \gls{ber} incurred by  Algorithm \ref{alg:memory} and the \gls{sct} decoder on a Ising, Trapdoor, ISI channel with $m=2$ and the MA-AGN channel. All figures compare with the ground truth given by the \gls{sct} decoder, except of the MA-AGN channel for which there is no known optimal polar decoder.}
     \label{fig:memory-symmetric}
\end{figure*}
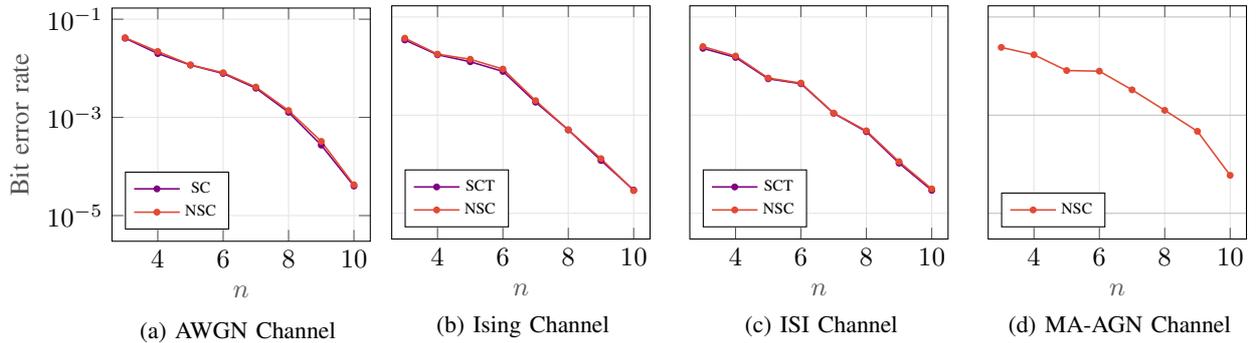
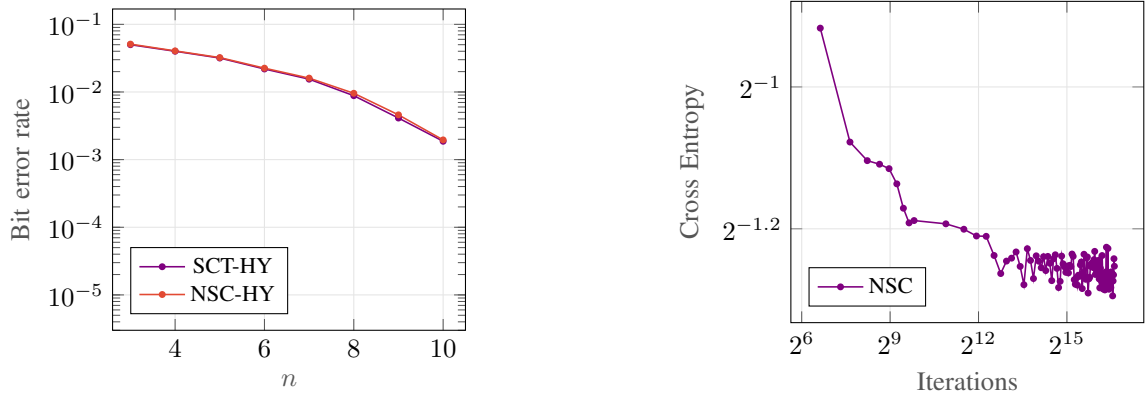
\begin{figure*}[t]
\def\x{0.05cm}
     \centering
     \hspace{-1.05cm}
     \begin{subfigure}[]{0.42\textwidth}
        \centering
\begin{tikzpicture}

\definecolor{chocolate2267451}{RGB}{226,74,51}
\definecolor{dimgray85}{RGB}{85,85,85}
\definecolor{gainsboro229}{RGB}{229,229,229}

\begin{axis}[
width=2.45in,
height=2.3in,
log basis y={10},
x grid style={gainsboro229},
xlabel=\textcolor{dimgray85}{\(\displaystyle n\)},
xmajorgrids,
xmin=2.6, xmax=10.5,
xtick style={color=dimgray85},
y grid style={gainsboro229},
ylabel=\textcolor{dimgray85}{Bit error rate},
ymajorgrids,
ymin=3e-06, ymax=0.168600008718341,
ymode=log,
ytick style={color=dimgray85},
legend style={at={(0.05,0.05)},anchor=south west, nodes={scale=0.85, transform shape} }
]
\addplot [semithick, violet, mark=*, mark size=1, mark options={solid}]
table {%
3	0.05
4	0.04
5	0.03175
6	0.021875
7	0.01546875
8	0.0088125
9	0.00411875
10	0.001868219
};
\addlegendentry{SCT-HY}

\addplot [semithick, chocolate2267451, mark=*, mark size=1, mark options={solid}]
table {%
3 0.051
4 0.0405
5 0.03225
6 0.0225
7 0.01603125
8  0.00956875
9 0.004590029761904762
10 0.0019566482843137254
};
\addlegendentry{NSC-HY}
\end{axis}

\end{tikzpicture}         
         \label{fig:memory-hy}
     \end{subfigure}
     \hspace{1cm}
     \begin{subfigure}[]{0.42\textwidth}
         \centering
\begin{tikzpicture}

\definecolor{chocolate2267451}{RGB}{226,74,51}
\definecolor{dimgray85}{RGB}{85,85,85}
\definecolor{gainsboro229}{RGB}{229,229,229}

\begin{axis}[
width=2.45in,
height=2.3in,
log basis y={2},
log basis x={2},
x grid style={gainsboro229},
xlabel=\textcolor{dimgray85}{Iterations},
xmajorgrids,
xtick style={color=dimgray85},
y grid style={gainsboro229},
ylabel=\textcolor{dimgray85}{Cross Entropy},
ymajorgrids,
ymode=log,
xmode=log,
legend style={at={(0.05,0.05)},anchor=south west, nodes={scale=0.85, transform shape} }
]
\addplot [semithick, violet, mark=*, mark size=1, mark options={solid}]
table{
99	0.529642678	
199	0.47383987	
299	0.465324502	
399	0.463704214	
499	0.461685277	
599	0.454882154	
699	0.444178408	
799	0.437864328	
899	0.43887722	
1899	0.437472106	
2899	0.43516274	
3899	0.432297876	
4899	0.432160182	
5899	0.424145555	
6899	0.416724347	
7899	0.421828784	
8899	0.423037063	
9899	0.425610548	
10899	0.419650427	
11899	0.412209415	
12899	0.426969064	
13899	0.422101489	
14899	0.414695843	
15899	0.424000552	
16899	0.421759603	
17899	0.419090245	
18899	0.423659054	
19899	0.417936223	
20899	0.42376031	
21899	0.420919723	
22899	0.41384255	
23899	0.422769188	
24899	0.424433128	
25899	0.418833584	
26899	0.41106292	
27899	0.41355074	
28899	0.42391105	
29899	0.420711293	
30899	0.419472936	
31899	0.417230755	
32899	0.41972149	
33899	0.416981887	
34899	0.42008562	
35899	0.419591454	
36899	0.424744047	
37899	0.423787573	
38899	0.414216568	
39899	0.412367931	
40899	0.415015089	
41899	0.412053564	
42899	0.416041233	
43899	0.415576776	
44899	0.420292887	
45899	0.42145776	
46899	0.41065532	
47899	0.417109317	
48899	0.424656019	
49899	0.415577724	
50899	0.414405893	
51899	0.423091145	
52899	0.423522241	
53899	0.408849813	
54899	0.414575397	
55899	0.415204832	
56899	0.419688526	
57899	0.420756973	
58899	0.420574084	
59899	0.421210875	
60899	0.419142651	
61899	0.422425002	
62899	0.425849412	
63899	0.416351134	
64899	0.420100296	
65899	0.417999268	
66899	0.417477355	
67899	0.420776005	
68899	0.422132142	
69899	0.414536559	
70899	0.410929616	
71899	0.418895075	
72899	0.416668815	
73899	0.413515516	
74899	0.42412465	
75899	0.416761288	
76899	0.410465542	
77899	0.41235192	
78899	0.41634759	
79899	0.410098793	
80899	0.416967249	
81899	0.417403301	
82899	0.42757837	
83899	0.412885291	
84899	0.4165782	
85899	0.42706384	
86899	0.41190721	
87899	0.410385896	
88899	0.414754392	
89899	0.415776444	
90899	0.416536448	
91899	0.41106091	
92899	0.417389859	
93899	0.414412001	
94899	0.416719292	
95899	0.407735722	
96899	0.413705775	
97899	0.416179366	
98899	0.422798983	
99899	0.419725663

};
\addlegendentry{NSC};
\end{axis}

\end{tikzpicture}         
         \label{fig:memory-convergence}
     \end{subfigure}
     \caption{The left figure illustrates the performance of Algorithm \ref{alg:memory-hy} on the Ising channel in comparison with the ground truth given by the Honda-Yamamoto scheme}
     \label{fig:memory-compare}
\end{figure*}

\subsection{Memoryless channels}\label{sec:exp:memoryless}
\par The following experiments test the proposed methodology to design polar codes for various memoryless channels. To demonstrate our algorithm, we conduct our experiments on both symmetric and non-symmetric \gls{bdmc}. The \gls{bsc} and the \gls{awgn} channels are chosen as instances of symmetric channels. A non symmetric \gls{bec}, as defined in \cite{honda2012polar}, is chosen as an instance of an asymmetric \glspl{bdmc}. To validate our numerical results, we compare our algorithms with the \gls{sc} decoder that provides the optimal decoding rule under the framework of polar codes.

\par The \gls{bsc} channel is defined by $W(y|x) = p\bI[y\neq x] + (1-p)\bI[y=x]$; here, we choose $p=0.1$. The \gls{awgn} channel is defined by the following relation $Y = X+N$, where $X$ is the channel input, $Y$ is the channel output and $N\sim \cN(0,\sigma^2)$ is an \gls{iid} Gaussian noise. In our experiments $\sigma^2=0.5$. The non-symmetric \gls{bec} is defined by two erasures probabilities, $\epsilon_0, \epsilon_1$, namely the probabilities for an erasure of the $``0"$ symbol and the $``1"$ symbol, respectively. Accordingly, $W(x|x)=1-\epsilon_x,\; W(?|x)=\epsilon_x$ for $x\in\{0,1\}$. Similar to \cite{honda2012polar}, we choose $\epsilon_0=0.4, \epsilon_1=0.8159$.
\begin{figure*}[b]
\def\x{0.05cm}
     \centering
     \hspace{-1.05cm}
     \begin{subfigure}[]{0.4\textwidth}
        \centering
\begin{tikzpicture}

\definecolor{chocolate2267451}{RGB}{226,74,51}
\definecolor{dimgray85}{RGB}{85,85,85}
\definecolor{gainsboro229}{RGB}{229,229,229}
\definecolor{darkblue}{RGB}{0,0,128}
\definecolor{darkgreen}{RGB}{0,128,0}
\definecolor{darkred}{RGB}{128,0,0}
\definecolor{darkpurple}{RGB}{128,0,128}
\definecolor{darkorange}{RGB}{255,140,0}

\begin{axis}[
width=2.45in,
height=2.3in,
log basis y={10},
x grid style={gainsboro229},
xlabel=\textcolor{dimgray85}{\(\displaystyle n\)},
xmajorgrids,
xmin=2.6, xmax=10.5,
xtick style={color=dimgray85},
y grid style={gainsboro229},
ylabel=\textcolor{dimgray85}{Bit error rate},
ymajorgrids,
ymin=3e-06, ymax=0.168600008718341,
ymode=log,
ytick style={color=dimgray85},
legend style={at={(0.05,0.05)},anchor=south west, nodes={scale=0.85, transform shape} }
]
\addplot [semithick, violet, mark=*, mark size=1, mark options={solid}]
table {%
3	0.037
4	0.015
5	0.017
6	0.005291667
7	0.003669643
8	0.000567568
9	0.000152483
10	1.69414E-05
};
\addlegendentry{$n_\mathsf{t}=10$}

\addplot [semithick, chocolate2267451, mark=*, mark size=1, mark options={solid}]
table {%
3	0.05
4	0.023
5	0.01475
6	0.0076875
7	0.004291667
8	0.000770579
9	0.000254112
10	5.49102E-05
};
\addlegendentry{$n_\mathsf{t}=9$}
\addplot [semithick, darkblue, mark=*, mark size=1, mark options={solid}]
table {%
3	0.048
4	0.018
5	0.01225
6	0.008
7	0.005054688
8	0.001525391
9	0.000497446
10	0.000196922
};
\addlegendentry{$n_\mathsf{t}=8$}
\addplot [semithick, darkorange, mark=*, mark size=1, mark options={solid}]
table {%
3	3.900e-02
4	2.475e-02
5	2.475e-02
6	1.281e-02
7	5.638e-03
8	3.081e-03
9	1.056e-03
10	4.430e-04
};
\addlegendentry{$n_\mathsf{t}=7$}

\addplot [semithick, darkred, mark=*, mark size=1, mark options={solid}]
table {%
3	0.036
4	0.0125
5	0.019
6	0.01275
7	0.008145833
8	0.003551136
9	0.001977214
10	0.001360085
};
\addlegendentry{$n_\mathsf{t}=6$}

\addplot [semithick, darkgreen, mark=*, mark size=1, mark options={solid}]
table {%
3	0.041
4	0.019
5	0.01825
6	0.0120625
7	0.007854167
8	0.004517857
9	0.002153556
10	0.001516778
};
\addlegendentry{$n_\mathsf{t}=5$}

\end{axis}

\end{tikzpicture}         
         \label{fig:nt-compare}
         \caption{Comparison of the incurred \glspl{ber} for varying values of $n_\mathsf{t}$ when applying Algorithm \ref{alg:memory} on the Ising channel.}
     \end{subfigure}
     \hspace{1cm}
     \begin{subfigure}[]{0.45\textwidth}
        \centering
        \begin{tabular}{|c|c|c|}
        \hline
        $m$ & SCT & NSC \\
        \hline
        1 & $O(2^3 N\log N)$   & $O( 800 N\log N)$ \\
        2 & $O(2^6 N\log N)$   & $O( 800 N\log N)$ \\
        3 & $O(2^9 N\log N)$   & $O( 800 N\log N)$ \\
        4 & $O(2^{12}N\log N)$ & $O( 800 N\log N)$ \\
        5 & $O(2^{15}N\log N)$ & $O( 800 N\log N)$ \\
        6 & $O(2^{18}N\log N)$ & $O( 800 N\log N)$ \\
        \hline
        \end{tabular}
        \caption{Comparison of the computational complexity of the \gls{sct} decoder and the \gls{nsc} decoder on the ISI channel for varying values of $m$. The complexity of the \gls{nsc} corresponds for $d=8$ and the \glspl{nn} have $50$ hidden units.}
        \label{tab:isi-complexity}
     \end{subfigure}
     \label{fig:isi-complexity}
\end{figure*}
\par Figure \ref{fig:memoryless-symmetric} shows the application of Algorithm \ref{alg:memoryless} on the \gls{bsc} and the \gls{awgn} channels. It reports the obtained \glspl{ber} by Algorithm \ref{alg:memoryless} in comparison with the \gls{sc} decoder. 
Figure \ref{fig:memoryless-non-symmetric} illustrates two comparisons. The first compares the \glspl{ber} obtained via the extension of Algorithm \ref{alg:memoryless} to the Honda-Yamamoto scheme, as desribed in Section \ref{sec:hy:memoryless}, and by the optimal decoding rule of the Honda-Yamamoto scheme. It also compares the \glspl{ber} obtained by a symmetric input distribution with the capacity achieving input distribution. In opposite to what we expected, better \glspl{ber} were obtained via a symmetric input distribution. The reason for this stems from the polarization of the source $U^N$, which has negative effect in short blocks.

\subsection{Channels with Memory}\label{sec:exp:memory}
\par The experiments here demonstrate the performance of the \gls{nsc} decoder on various channels. First, having a sample $\cD_{M,N}$ does not indicate if it is drawn from channel with or without memory. As memoryless channels are a special case of channels with memory, we start by testing the \gls{nsc} on memoryless channels. The experiments proceed to \glspl{fsc} for which there exists an analytic polar decoder that is given by the \gls{sct} decoder \cite{wang2015construction}. Recall that the computational complexity of the \gls{sct} decoder is $O(|\cS|^3 N \log N)$; therefore, we evaluate our algorithms on channels with a small state space and on channels with a large state space, i.e. $|\cS|^3\gg N\log N$. 
The last experiments test the \gls{nsc} decoder on channels with infinite state space for which an optimal decision rule is intractable.

\par As instances of channels with memory, we choose the Ising channel \cite{ising1925beitrag} and the \gls{isi} channel \cite{cover2006elements}, respectively. These channels
belong to the family of \glspl{fsc}, and therefore, their optimal decoding rule is given by the \gls{sct} decoder. We also tested our methodology on channels with continuous state space for which the \gls{sct} decoder can not be applied. As an instance of such channels, we choose the \gls{maagn} channel. 
The Ising channel \cite{ising1925beitrag} is defined by $Y=X$ or $Y=S$ with equal probability, and $S^\prime = X$,
where $X$ is the channel input, $Y$ is the channel output, $S$ is the channel states in the beginning of the transmission and $S^\prime$ is the channel's state in the end of the transmission. 
The interference channel is defined by the formula $Y_t = \sum_{i=0}^m h_i X_{t-i} + Z_i$,
where $X_t, Y_t$ are the channel input and output at time $t$, $\{h_i\}_{i=1}^m$ are the interference parameters and $Z_i\stackrel{iid}{\sim}\cN(0,\sigma^2)$. In our experiment we set $h_i=0.9^i$ and $\sigma^2=0.5$.
The \gls{maagn} channel is given by $Y_t=X_t+\tilde{Z}_t$, $\tilde{Z}_t=Z_t + \alpha Z_{t-1}$,
where $\alpha\in\bR$ and $Z_t\stackrel{iid}{\sim}\cN(0,\sigma^2)$.

\par Figure \ref{fig:memory-symmetric} compares the \glspl{ber} attained by Algorithm \ref{alg:memory} vs. the optimal decoding rule for the \gls{awgn}, Ising, ISI with $m=2$ and the \gls{maagn} channel with $\alpha=0.9$. For the last channel, we illustrate the incurred \glspl{ber} without any comparison, since, as far as we know, there is no available decoding rule for channel with continuous state space. 
Figure \ref{fig:memory-compare} illustrates the the \glspl{ber} attained by Algorithm \ref{alg:memory-hy} on the Ising channel with a stationary (non-symmetric) input distribution, the convergence of Algorithm \ref{alg:memory} when applied on the Ising channel. It also illustrates the \glspl{ber} incurred for a varying values of $n_\mathsf{t}$, the block length in training. It is clear that increasing $n_\mathsf{t}$ produce better estimation results. Table \ref{tab:isi-complexity} depicts the decoding complexity of the \gls{sct} and the \gls{nsc} decoders for the \gls{isi} channel. 

\section{Proofs}\label{sec:proofs}
This section provides the proofs for the main theorems of the paper.
\subsection{Proof of Theorem \ref{thm:mine_cons}}
The proof relies on the consistency of the \gls{mine} algorithm \cite[Theorem 2]{belghazi2018mine}. Specifically, for all $\varepsilon>0$ there exists $m>0$ such that for all $M>m$ $\bP-a.s.$ $$\left| \f{\sI_\Phi}{\cD_{M}}-\f{\sI}{X;Y}\right| <\varepsilon,$$
where $\f{\sI_\Phi}{\cD_{M}}=\max_{\phi\in\Phi} \frac{1}{M}\sum_{i=1}^{M} \f{T_\phi}{x_i,y_i}-\log \frac{1}{M}\sum_{i=1}^{M} e^{\f{T_\phi}{x_i,\tilde{y}_i}}$, as defined in Equation \eqref{eqn:mine}. Recall that $T_\Phi$ denotes the maximizer of the \gls{mine} algorithm and $T^\ast=\log\frac{P_{X,Y}}{P_X\otimes P_Y}$. 

\par The next argument claims that for all $\varepsilon^\prime>0$ there exists $\varepsilon>0$ such that $\bP-a.s.$ $\left|T^\ast-T_\Phi\right|<\varepsilon^\prime$. 
For this we define the Gibbs density $G_\Phi \triangleq \frac{P_X\otimes P_Ye^{T_\Phi}}{\Ep{P_X\otimes P_Y}{e^{T_\Phi}}}$. First note that $\Ep{P_{X,Y}}{G_\Phi}=1$ and $G_\Phi \ge 0$ and therefore is a valid density. For $P\triangleq P_{X,Y}$ and $Q\triangleq P_X\otimes P_Y$, we observe
\begin{align}
\kl{P}{G_\Phi} &=\Ep{P}{\log\frac{P}{\frac{Qe^{T_\Phi}}{\Ep{Q}{e^{T_\Phi}}}}} \\ 
&=\Ep{P}{\log\frac{P}{Q} - \log \frac{e^{T_\Phi}}{\Ep{Q}{e^{T_\Phi}}}} \\
&=\f{\sI}{X;Y} - \left(\Ep{P}{T_\Phi}-\log\Ep{Q}{e^{T_\Phi}}\right)  \ge  0,
\end{align}
with equality if and only if $P=G_\Phi$ almost everywhere. From the consistency of \gls{mine} we have that $\kl{P}{G_\Phi}<\varepsilon$. Since the mappings $T_\Phi \mapsto G_\Phi$ and $G_\Phi\mapsto\kl{P}{G_\Phi}$ are continuous, we conclude that $\kl{P}{G_\Phi}<\varepsilon$ implies $\bP-a.s.$ $\left|T^\ast-T_\Phi\right|<\varepsilon^\prime$.

\par The last step is to observe that the mapping of the \gls{sc} decoder is continuous. Specifically, the mapping $E^W_\Phi = \f{T_\Phi}{1, \cdot} - \f{T_\Phi}{0, \cdot}$ is continuous and $F,G,H$, as defined in Equation \eqref{eqn:sc_ops}, are continuous, and therefore for all $\varepsilon^{\prime\prime}>0$ there exists $\varepsilon^\prime>0$ such that $\left| \f{\sH_{\Phi}^M}{U_i\vert U^{i-1}, Y^N} -\f{\sH}{U_i\vert U^{i-1}, Y^N}\right| <\varepsilon^{\prime\prime}$, which concludes the proof.
\subsection{Proof of Theorem \ref{thm:nsc}}
    \par The proof starts with identifying the structure of the \gls{nsc}. Specifically, it concludes there exists recursive formulas for \gls{sc} decoding of the polar code, for which the channel embedding belong to a $d$-dimensional Euclidean space, for some $d\in \bN$. First, we define the channel embedding, as defined by the authors of \cite{wang2015construction}.
    \begin{definition}[Channel Embedding]
        Let $W$ be a \gls{fsc}. Then, $\f{W^{(i)}_N}{y_1^N,u_1^{i-1}}$ is defined by 
        \begin{align}
          \f{W^{(i)}_N}{y_1^N,u_1^{i-1}} = \left\{{\left(s_0,s_N,u_{i}\right)\in\cS\times\cS\times\cX}\;:\;\;\f{P_{Y_1^N,U_1^i,S_N|S_0}}{y_1^N,u_1^i,s_N \vert s_0}\right\}.
        \end{align}
    \end{definition}
    \noindent The channel embedding is a set of all the joint probabilities of $y_1^N,u_1^i$ given every initial state and final state of the channel. The channel embedding $\f{W^{(i)}_N}{y_1^N,u_1^{i-1}}$ is a finite set if $|\cX|,|\cS|<\infty$. It is a deterministic function of $Y_1^N,U_1^{i-1}$ and it represented by a vector whose dimension is $|\cS|^2(|\cX|-1)$; its entries correspond to all possible values of $U_i,S_0,S_N$. Herein, the dimension of the vector is denoted by $|\cS|^2$ due to the assumption that $|\cX|=2$. 

    The next lemma shows the structure of the polar decoder, for which the channel embeddings are represented in $\bR^d$ for any $d\in\bN$, instead of $\bR^{|\cS|^2}$ as in the case of the \gls{sct} decoder. 
    \begin{lemma}[Structure]\label{lemma:nsc_structure}
        Fix $n\in\bN$ and set $N=2^n, i\in[N]$. Let $\cE \subset \bR^d$ for $d\in\bN$. Then, for all $i\in[N]$ the functions $P_{U_i|U_1^{i-1},Y_1^N}$ may be computed using four distinct functions $E: \cY\to\cE, F:\cE\times\cE\to\cE, G:\cE\times\cE\times\cX\to\cE,H:\cE\to[0,1]$. The functions are given by 
        \begin{align}
            \f{E^{(1)}_1}{y_i} &= \f{E}{y_i} \\
            \f{E^{(2i-1)}_{2N}}{y_1^{2N}, u_1^{2i-2}} &= \f{F}{\f{E^{(i)}_{N}}{y_1^{N},u_{1,e}^{2i-2}\oplus u_{1,o}^{2i-2}}, \f{E^{(i)}_{N}}{y_{N+1}^{2N},u_{1,e}^{2i-2}}}  \\
            \f{E^{(2i)}_{2N}}{y_1^{2N}, u_1^{2i-1}} &= \f{G}{\f{E^{(i)}_{N}}{y_1^{N},u_{1,e}^{2i-2}\oplus u_{1,o}^{2i-2}}, \f{E^{(i)}_{N}}{y_{N+1}^{2N},u_{1,e}^{2i-2}}, u_{2i-1}}  \\
            \f{P_{U_i|U^{i-1},Y^N}}{1|u_1^{i-1},y_1^{N}} &= \f{H}{\f{E^{(i)}_{N}}{y_1^{N},u_{1}^{i-1}}}.
        \end{align}
    \end{lemma}

 \par Let $E_\phi^W$ and $F_\theta, G_\theta, H_\theta$, as defined in Definition \ref{def:nsc}, be \glspl{nn}. For simplicity we denote $\Theta^\prime = \Theta \times \Phi$. For every choice of $\theta\in\Theta^\prime$, there is an induced distribution denoted by $P^\theta_{U_i\vert U^{i-1}, Y^N}$; we denote by $\f{P_\theta}{U^{i-1},Y^N} \triangleq \f{P^\theta_{U_i\vert U^{i-1}, Y^N}}{1\vert U^{i-1},Y^N}$.
 The following lemma states that there exists $\theta\in\Theta^\prime$ that can approximate $\f{\sH}{U_i\vert U^{i-1}, Y^N}$ with an arbitrary precision. 
    \begin{lemma}[Approximation]\label{lemma:end2end_approx}
        Let $\varepsilon>0$. There exists a \gls{nn} $P_\theta: \cX^{i-1}\times \cY^N \to [0,1]$, with parameters $\theta$ in some compact space $\Theta^\prime\subset \bR^p$, $p\in\bN$, such that 
        \begin{equation}\label{eqn:end2end_approx}
            \Big\vert \f{\sH}{U_i\vert U^{i-1}, Y^N} - \f{\sH_{\Theta^\prime}}{U_i\vert U^{i-1}, Y^N} \Big\vert < \varepsilon,
        \end{equation}
        where 
        \begin{equation}
            \f{\sH_{\Theta^\prime}}{U_i\vert U^{i-1}, Y^N} = \min_{\theta\in\Theta^\prime}\E{-U_i\log\f{ P_\theta}{U^{i},Y^N}-(1-U_i)\log\left(1-\f{P_\theta}{U^{i},Y^N}\right)}.
        \end{equation}
    \end{lemma}
    \par The third lemma guarantees the almost sure convergence of the empirical average to the expected value when the number of samples goes to infinity.
    \begin{lemma}[Estimation]\label{lemma:end2end_est}
        Let $\varepsilon>0$. Let $P_\theta: \cX^{i-1}\times \cY^N \to [0,1]$, with parameters $\theta$ in some compact space $\Theta^\prime\subset \bR^p$, $p\in\bN$. Then, there exists $m\in\bN$ such that for all $M>m$, $\bP \text{ a.s.}$,
        \begin{equation}\label{eqn:end2end_est}
            \Big\vert \f{\sH_{\Theta^\prime}}{U_i\vert U^{i-1}, Y^N} - \f{\sH^M_{\Theta^\prime}}{U_i\vert U^{i-1}, Y^N} \Big\vert < \varepsilon,
        \end{equation}
                where 
        \begin{equation}
            \f{\sH^M_{\Theta^\prime}}{U_i\vert U^{i-1}, Y^N} = \min_{\theta\in\Theta^\prime}\frac{1}{M}\sum_{j=1}^M \left\{-U_i\log\f{ P_\theta}{U^{i},Y^N}-(1-U_i)\log\left(1-\f{P_\theta}{U^{i},Y^N}\right)\right\}.
        \end{equation}
    \end{lemma}

    \par The proof is concluded by the combination of Lemmas \ref{lemma:end2end_approx} and \ref{lemma:end2end_est}. Specifically, we use the lemmas to claim that there exist $\Theta^\prime\in\bR^p$ and $m\in\bN$ such the Equations \eqref{eqn:end2end_est} and \eqref{eqn:end2end_approx} hold with $\frac{\varepsilon}{2}$. By the triange inequality, for all $M>m$, $\bP-a.s.$
    \begin{align}
        &\left| \f{\sH_{\Theta^\prime}^M}{U_i\vert U^{i-1}, Y^N} -\f{\sH}{U_i\vert U^{i-1}, Y^N}\right|  \nn\\
        &\leq\left| \f{\sH_{\Theta^\prime}}{U_i\vert U^{i-1}, Y^N} -\f{\sH_{\Theta^\prime}^M}{U_i\vert U^{i-1}, Y^N}\right|+\left| \f{\sH_{\Theta^\prime}}{U_i\vert U^{i-1}, Y^N} -\f{\sH}{U_i\vert U^{i-1}, Y^N}\right|\nn\\&<\varepsilon,
    \end{align}
    which concludes the proof.
    \subsubsection{Proof of Lemma \ref{lemma:nsc_structure}}
    \par We derive the structure of $P_{U_i|U^{i-1},Y^N}$ by utilizing the mechanism of the \gls{sc} decoder. The following derivation is similar to the one in the proof of \cite[Theorem 2]{wang2015construction}. According to \cite[Theorem 2]{wang2015construction}, the channel embedding is given by
        \begin{align}
             \f{E^{(i)}_N}{y_1^N,u_1^{i-1}} &= \left\{\f{P_{Y_1^N,U_1^i,S_N|S_0}}{y_1^N,u_1^i,s_N \vert s_0}\right\}_{\left(s_0,s_N,u_{i}\right)\in\cS\times\cS\times\cX}. \label{eqn:trellis}
        \end{align}
    For every $s_0,s_N,u_{i}\in\cS\times\cS\times\cX$, it was shown in \cite{wang2015construction} that the recursive formulas are given by
        \begin{align}
            &\f{P_{Y_1^{2N},U_1^{2i-1},S_{2N}|S_0}}{y_1^{2N},u_1^{2i-1},s_{2N}|s_0} \nn\\ &=\sum_{s^\prime_N}\sum_{u_{2i}}\f{P_{Y_1^N,U_1^i,S_N|S_0}}{y_1^N,u_{1,e}^{2i}\oplus u_{1,o}^{2i},s^\prime_{N} \vert s_0} \cdot \f{P_{Y_1^N,U_1^i,S_N|S_0}}{y_{N+1}^{2N},u_{1,e}^{2i},s_{2N} \vert s^\prime_N }, \label{eqn:sct_check}\\
            &\f{P_{Y_1^{2N},U_1^{2i},S_{2N}|S_0}}{y_1^{2N},u_1^{2i},s_{2N}|s_0} \nn\\
            &=\sum_{s^\prime_N}\f{P_{Y_1^N,U_1^i,S_N|S_0}}{y_1^N,u_{1,e}^{2i}\oplus u_{1,o}^{2i},s^\prime_{N} \vert s_0} \cdot \f{P_{Y_1^N,U_1^i,S_N|S_0}}{y_{N+1}^{2N},u_{1,e}^{2i},s_{2N} \vert s^\prime_N }. \label{eqn:sct_bit} 
        \end{align}
    The soft decision is given by 
    \begin{align}
 \f{P_{U_i\vert U^{i-1},Y^N}}{1\vert u_{1}^{i-1},y^{N}} &= \f{\sigma}{\log\frac{\sum_{s_N,s_{0}} \f{P_{U_i,U^{i-1},Y^N,S_N|S_0}}{1,u_{1}^{i-1},y^{N},s_N|s_0}P_{S_0}(s_0)}{\sum_{s_N,s_0}\f{P_{U_i,U^{i-1},Y^N,S_N|S_0}}{0,u_{1}^{i-1},y^{N},s_N|s_0}P_{S_0}(s_0)}} \label{eqn:sct_soft},
    \end{align}
    where $\sigma$ denotes the logistic (sigmoid) function. Equations  \eqref{eqn:sct_check},\eqref{eqn:sct_bit},\eqref{eqn:sct_soft} can be rewritten in a more compact form in terms of $\f{W^{(i)}_N}{y_1^N,u_1^{i-1}}$, i.e. 
        \begin{align}
            \f{E^{(2i-1)}_{2N}}{y_1^{2N}, u_1^{2i-2}} &= \f{\widetilde{F}}{\f{E^{(i)}_{N}}{y_1^{N},u_{1,e}^{2i-2}\oplus u_{1,o}^{2i-2}}, \f{E^{(i)}_{N}}{y_{N+1}^{2N},u_{1,e}^{2i-2}}}  \\
            \f{E^{(2i)}_{2N}}{y_1^{2N}, u_1^{2i-1}} &= \f{\widetilde{G}}{\f{E^{(i)}_{N}}{y_1^{N},u_{1,e}^{2i-2}\oplus u_{1,o}^{2i-2}}, \f{E^{(i)}_{N}}{y_{N+1}^{2N},u_{1,e}^{2i-2}}, u_{2i-1}}  \\
            \f{P_{U_i|U^{i-1},Y^N}}{1|u_1^{i-1},y_1^{N}} &= \f{\widetilde{H}}{\f{E^{(i)}_{N}}{y_1^{N},u_{1}^{i-1}}},
        \end{align}
    where $\widetilde{F}$ denotes the check-node, $\widetilde{G}$ denotes the bit-node and $\widetilde{H}$ denotes the soft decision. Since $W^{(i)}_{N}\in \bR^{|\cS|^2}$, we can find a bijection between $\bR^{|\cS|^2}$ and $\bR^{d}$ for any $d\in\bN$ (since $|\bR^d|=|\bR|$). We denote this mapping by $K:\bR^{|\cS|^2}\to\bR^{d}$, and its inverse by $K^{-1}$. Using $K$, we may rewrite the equations in the following way
    \begin{align}
        \f{E^{(1)}_1}{y_i} &= \f{K}{\f{W^{(1)}_1}{y_i}} \\
        \f{E^{(2i-1)}_{2N}}{y_1^{2N}, u_1^{2i-2}} &= \f{K}{\f{\widetilde{F}}{\f{K^{-1}}{\f{E^{(i)}_{N}}{y_1^{N},u_{1,e}^{2i-2}\oplus u_{1,o}^{2i-2}}}, \f{K^{-1}}{\f{E^{(i)}_{N}}{y_{N+1}^{2N},u_{1,e}^{2i-2}}}}}  \label{eqn:nsc_check}\\
        \f{E^{(2i)}_{2N}}{y_1^{2N}, u_1^{2i-1}} &= \f{K}{\f{\widetilde{G}}{\f{K^{-1}}{\f{E^{(i)}_{N}}{y_1^{N},u_{1,e}^{2i-2}\oplus u_{1,o}^{2i-2}}}, \f{K^{-1}}{\f{E^{(i)}_{N}}{y_{N+1}^{2N},u_{1,e}^{2i-2}}}, u_{2i-1}}}  \label{eqn:nsc_bit}\\
        \f{P_{U_i|U^{i-1},Y^N}}{1|u_1^{i-1},y_1^{N}} &= \f{\widetilde{H}}{\f{K^{-1}}{\f{E^{(i)}_{N}}{y_1^{N},u_{1}^{i-1}}}}. \label{eqn:nsc_soft}
    \end{align}
    We denote by $E=K^{-1}\circ W^{(1)}_1:\cY\to\cE$ the channel embedding, and Equations \eqref{eqn:nsc_check}, \eqref{eqn:nsc_bit}, \eqref{eqn:nsc_soft}, define $F,G,H$, as requested.
    Thus, we showed that $\f{P_{U_i\vert U^{i-1},Y^N}}{1\vert u_{1}^{i-1},y^{N}}$ is computed by four distinct functions, which concludes the proof.

    \subsubsection{Proof of Lemma \ref{lemma:end2end_approx}}
    \par Let $L^\ast= -\log P_{U_i\vert U^{i-1},Y^N}$ and $L_\theta = -U_i\log\f{ P_\theta}{U^{i},Y^N}-(1-U_i)\log\left(1-\f{P_\theta}{U^{i},Y^N}\right)$. By definition, $L_\theta$ is bounded by a constant $R^\prime$.
    By construction $\E{L^\ast} = \f{\sH}{U_i\vert U^{i-1}, Y^N}$. Fix $\varepsilon>0$, and assume $L^\ast$ is bounded by a constant $R^{\prime\prime}>0$. Let $R=\max\{R^{\prime},R^{\prime\prime}\}$.
    By the universal approximation theorem \cite{schafer2006recurrent}, 
    there exists $\hat{\theta}\in\Theta^\prime$ with $L_{\hat{\theta}} \leq R$ such that 
    \begin{equation}\label{eqn:proof:approx:bounded}
        \left|L^\ast-L_{\hat{\theta}}\right| < \frac{\varepsilon}{2}.
    \end{equation}
    For the case where $L^\ast$ is unbounded, we first note that $L^\ast$ is integrable since
    \begin{align*}
        \E{\left|L^\ast\right|} = \E{-\log P_{U_i|U^{i-1},Y^N}} \leq H(U_i) \leq 1.
    \end{align*}
    Therefore, by the dominated convergence theorem, there exists $R>0$ such that 
    \begin{equation}\label{eqn:proof:approx:unbounded}
        \left|L^\ast \bI_{L^\ast > R}\right| < \frac{\varepsilon}{2}.
    \end{equation}
    Combining Equations \eqref{eqn:proof:approx:bounded}, \eqref{eqn:proof:approx:unbounded} and the triangle inequality yields 
    \begin{align*}
        \E{\left|L^\ast-L_{\hat{\theta}}\right|} \leq \E{\left|L^\ast-L_{\hat{\theta}}\right|\bI_{L^\ast \leq R}}+\E{\left|L_{\hat{\theta}}\right|\bI_{L^\ast > R}} +\E{\left|L^\ast\right|\bI_{L^\ast > R}} \leq \frac{\varepsilon}{2}+0+\frac{\varepsilon}{2}=\varepsilon.
    \end{align*}
    This implies that the difference between $\f{\sH_{\Theta^\prime}}{U_i\vert U^{i-1}, Y^N},\f{\sH^M_{\Theta^\prime}}{U_i\vert U^{i-1}, Y^N}$ is at most $\varepsilon$.

    \subsubsection{Proof of Lemma \ref{lemma:end2end_est}}
    According to the triangle inequality
\begin{align*}
     \Big\vert \f{\sH_{\Theta^\prime}}{U_i\vert U^{i-1}, Y^N} - \f{\sH^M_{\Theta^\prime}}{U_i\vert U^{i-1}, Y^N} \Big\vert \leq \sup_{\theta\in\Theta^\prime} \Big\vert\E{L_\theta} - \frac{1}{M}\sum_{j=1}^M L_\theta \Big\vert.
\end{align*}
    Since $\Theta^\prime$ is compact and $L_\theta$ is continuous, the family of functions $L_\theta$ satisfy the uniform law of large numbers for stationary and ergodic processes \cite[Theorem 3.1]{pevskir1996uniform}. Therefore given $\varepsilon >0$, there exists $m\in\bN$ such that for all $M>m$ $\bP-a.s.$
    \begin{align*}
\Big\vert \f{\sH_{\Theta^\prime}}{U_i\vert U^{i-1}, Y^N} - \f{\sH^M_{\Theta^\prime}}{U_i\vert U^{i-1}, Y^N} \Big\vert < \frac{\varepsilon}{2}.
    \end{align*}

\subsection{Proof of Theorem \ref{thm:channel_sufficient}}
The first Markov relation is straightforward as $T^N$ is a function of $Y^N$ and $U^N=X^N G_N$. The second Markov relation is derived by showing that $T^N$ is a sufficient statistic of $Y^N$ for the estimation of $X^N$, or equivalently, $\f{I}{X^N;Y^N} = \f{I}{X^N;T^N}$. For $x^N,y^N\in \cX^N\times\cY^N$, consider the following chain of equalities:
\begin{align*}
    \f{P_{X^N, Y^N}}{x^n,y^n} 
    &= \prod_{i=1}^N \f{P_{X_i,Y_i\vert X^{i-1}, Y^{i-1}}}{x_i,y_i\vert x^{i-1},y^{i-1}} \\
    &\hspace{0cm}\stackrel{(a)}{=}\prod_{i=1}^N \f{P_{X_i\vert X^{i-1}}}{x_i\vert x^{i-1}}\f{P_Z}{y_i}
                     \frac{\f{P_{Y_i\vert X^{i}, Y^{i-1}}}{y_i\vert x^{i},y^{i-1}}}{\f{P_Z}{y_i}}\\
    &\hspace{0cm}\stackrel{(b)}{=}\f{\exp}{\log\f{P_Z^{\otimes N}}{y^N}}\f{\exp}{\sum_{i=1}^N \log\f{P_{X_i\vert X^{i-1}}}{x_i\vert x^{i-1}}} \\
    &\hspace{0.5cm}\times\f{\exp}{\sum_{i=1}^N\log\frac{\f{P_{Y_i\vert X^{i}, Y^{i-1}}}{y_i\vert x^{i},y^{i-1}}}{\f{P_Z}{y_i}}},
\end{align*}
where (a) follows from the chain rule, the absence of outputs feedback, and $P_{Y_i\vert X^{i}, Y^{i-1}}\ll P_Z$; and (b) is a result of rearranging the terms into exponents. Next, we identify that 
$\f{P_{X^N, Y^N}}{x^n,y^n} = \f{h}{y^N}\f{g}{\f{t^N}{y^N}, x^N}$, where
\begin{align*}
    \f{h}{y^N} &\triangleq \f{\exp}{\log\f{P_Z^{\otimes N}}{y^N}},\\
    \f{g}{t^N, x^N} &\triangleq \f{\exp}{\sum_{i=1}^N \log\f{P_{X_i\vert X^{i-1}}}{x_i\vert x^{i-1}}+ \f{t_{y^i}}{x^i}}.
\end{align*}
This is exactly the factorization in the well-known Fisher–Neyman factorization theorem \cite{fisher1935logic,neyman1935statistical}, and thus $T^N$ is a sufficient statistic of $Y^N$ for the estimation of $X^N$. Since $U^N=X^N G_N$ is bijective, we conclude the theorem.



\section{Conclusions and Future Work}\label{sec:conc}
\par This paper presents data-driven algorithms for the design of polar codes. It began with devising algorithms for symmetric input distributions and for memoryless channels. Next, it addressed the case of channels with memory. Here, we devised a \gls{nsc} decoder that approximated the \gls{sc} decoder core elements with \glspl{nn}. We showed that the \gls{nsc} decoder is consistent, in the sense for sufficiently many samples of the channel inputs and outputs, the symmetric capacity of the effective bit channels is estimated with an arbitrary precision. Next, we extended the methods for the case where the input distribution is not necessarily symmetric via the Honda-Yamamoto scheme.

\par We also showed the role of neural estimation methods in a data-driven design of polar codes. Specifically, for the case of memoryless channels the \gls{mine} algorithm provides a valid construction of the channel embedding. For channels with memory, the \gls{dine} algorithm may be used to construct the channel embedding. The importance of neural estimation methods comes from the fact that the channel embedding may be computed independently from the \gls{nsc} decoder. It is visible in the case of memoryless channels; indeed, in this case the channel embedding is the only thing need to be estimated and the \gls{sc} decoder for memoryless channels may be applied ``as is". This set one of the future research goals of this work -- to learn a common \gls{nsc} for all channels, and devise an algorithm to estimate a channel embedding that is valid for the common \gls{nsc}.

\par We demonstrated our approach on both memoryless channels and channels with memory. We showed that the proposed algorithms have similar decoding errors as their analytic counterpart, whenever they exists. We also demonstrated our algorithm on channels for which there is no practical \gls{sc} decoder in the literature. Another experiment is conducted on the \gls{isi} channel that emphasized its main advantage over the \gls{sct} decoder; its computational complexity does not grow with the state alphabet of the channel. Our next steps would be to examine our methodology with existing algorithms in the fields, such as list decoding \cite{tal2015list}. We are also interested in using polar coding framework for the problem of capacity estimation. This way, we will be able to estimate both the capacity and a capacity achieving code with the same algorithm.

\printbibliography

{
}
\end{document}